\documentclass{article}

\usepackage[ruled]{algorithm2e} 

\SetAlFnt{\small}
\SetAlCapFnt{\small}
\SetAlCapNameFnt{\small}
\SetAlCapHSkip{0pt}
\IncMargin{-\parindent}

\usepackage{booktabs}
\usepackage{makecell}
\usepackage{hyperref}
\usepackage{bbding}
\usepackage[nointegrals]{wasysym}
\usepackage[utf8]{inputenc}
\usepackage{xspace}
\usepackage{mathtools}
\usepackage{url}
\usepackage{subcaption}
\usepackage{tikz}
\usetikzlibrary{arrows, fit, positioning}
\usetikzlibrary{backgrounds,automata,calc}
\usetikzlibrary{decorations.pathreplacing}
\usetikzlibrary{shapes.arrows}
\tikzstyle{agentnode} = [draw,circle,fill=blue!5!white,inner sep = 0.1cm,minimum height=0.8cm,minimum width = 0.8cm]
\tikzstyle{itemnode} = [draw,fill=red!5!white,minimum height=0.8cm,minimum width = 0.8cm]
\usepackage{csquotes}
\usepackage{amsmath,amsfonts,amsthm}
\usepackage{bm,bbm}
\usepackage{pgfplots}
\pgfplotsset{compat=1.17}

\usepackage{enumerate}
\usepackage{paralist}
\usepackage[shortlabels]{enumitem}
\usepackage[margin = 1.25in]{geometry}
\usepackage{appendix}
\usepackage[capitalize]{cleveref}
\usepackage{thmtools}
\usepackage{mathtools}
\usepackage[numbers]{natbib}
\usepackage{changepage}

\usepackage{commath}
\usepackage{units}

\newcommand{\NSW}{\text{NSW}\xspace}

\newcommand{\R}{\mathbb{R}}

\renewcommand{\cal}[1]{\mathcal{#1}}

\newtheorem{theorem}{Theorem}[section]
\newtheorem{conjecture}[theorem]{Conjecture}
\newtheorem{lemma}[theorem]{Lemma}

\newtheorem{property}{Property}
\newtheorem*{theorem*}{Theorem}

\newcommand{\p}{{\rm P}\xspace}
\newcommand{\np}{{\rm NP}\xspace}
\newcommand{\apx}{{\rm APX}\xspace}
\newcommand{\Pos}{\textit{pos}\xspace}
\newcommand{\Neg}{\textit{neg}\xspace}
\newcommand{\clog}{\textit{clog}\xspace}
\newcommand{\poly}{\textit{poly}\xspace}
\usepackage{mathtools}

\setlist{topsep=0.5ex,itemsep=0.1ex}

\theoremstyle{definition}

\theoremstyle{definition}
\newtheorem{definition}[theorem]{Definition}

\newcommand{\prob}[3]{
\begin{description}
  \item[Name:] #1
  \item[Given:] #2
  \item[Question:] #3
\end{description}}

\title{On the Hardness of Fair Allocation under Ternary Valuations} 

\author{Zack Fitzsimmons \\ 
College of the Holy Cross \\ 
\texttt{zfitzsim@holycross.edu}\\
\and 
Vignesh Viswanathan \\
University of Massachusetts, Amherst\\
\texttt{vviswanathan@umass.edu} \\
\and
Yair Zick \\
University of Massachusetts Amherst\\
\texttt{yzick@umass.edu} \\
  }
\date{}

\begin{document}

\maketitle

\begin{abstract}
We study the problem of fair allocation of indivisible items when agents have ternary additive valuations --- each agent values each item at some fixed integer values $a$, $b$, or $c$ that are common to all agents. 
The notions of fairness we consider are max Nash welfare (MNW), when $a$, $b$, and $c$ are non-negative, and max egalitarian welfare (MEW). 
We show that for {\em any} distinct non-negative $a$, $b$, and $c$, maximizing Nash welfare is APX-hard --- i.e., the problem does not admit a PTAS unless $\p = \np$. 
We also show that for any distinct $a$, $b$, and $c$, maximizing egalitarian welfare is APX-hard except for a few cases when $b = 0$ that admit efficient algorithms. 
These results make significant progress towards completely characterizing the complexity of computing exact MNW allocations and MEW allocations. 
En route, we resolve open questions left by prior work regarding the complexity of computing MNW allocations under bivalued valuations, and MEW allocations under ternary mixed manna.
\end{abstract}

\section{Introduction}\label{sec:intro}

Fair allocation of indivisible items is a fundamental problem in computational social choice. 
We are given a set of indivisible \emph{items} that need to be distributed among \emph{agents} that have subjective \emph{valuations} for the items they receive. Many problems can be naturally cast as instances of the fair allocation problem. For example, one might wish to distribute a set of course seats to students, or schedule shifts to hospital workers. 
Our objective is to find an \emph{allocation} of items to agents satisfying certain natural \emph{justice criteria}. 
Unfortunately, when agents have arbitrary combinatorial valuations, several allocation desiderata are either computationally intractable to compute, or simply not guaranteed to exist (see e.g., \citet{Caragiannis2016MNW,pla-rou:j:efx}
as well as \citet{Aziz2022fairsurvey} for an overview). 

Thus, recent works study simpler classes of valuations where exact fair allocations can be computed. 
There is an efficient algorithm to compute Max Nash welfare allocations when agents have binary valuations \citep{barman2018pathtransfers,halpern2020binaryadditive,Babaioff2021Dichotomous}, i.e., where each item is valued at either $0$ or $1$. 
This result was later extended to bivalued additive valuations where each item is valued at $1$ or $c$ with $c$ being either an integer or a half-integer \citep{akrami2022halfintegers,akrami2022mnw}. 
Similarly, for the problem of allocating chores, there exists an efficient algorithm that computes leximin allocations when agents have binary costs \citep{barman2023chores} or bivalued costs when the ratio of the costs is $2$ \citep{lenstra1990chores}. 

Other works restrict their attention to bivalued instances in the realm of goods \citep{garg2021bivaluedadditive,amantidis2021mnw} as well as chores \citep{ebadian2022bivaluedchores,garg2022bivaluedchores,aziz2023twotypes,chakrabarty2015bivalued}. Generalizing beyond bivalued instances, much less is known about the complexity of fair allocation under ternary or trivalued instances: when each item is valued at $a$, $b$, or $c$ for some integers $a$, $b$, and $c$. For example, we do not know if an exact max Nash welfare allocation is efficiently computable when each item is valued at $0$, $1$, or $2$. Our goal in this paper is to bridge this gap by answering the following question:

\begin{quote}
\emph{What is the computational complexity of computing fair allocations under ternary valuations?} 
\end{quote}

\subsection{Our Results}

\begin{table*}[ht]
    \centering
    \begin{subtable}[c]{\linewidth}
        \centering
        \begin{tabular}{cc}
            \toprule
            Valuation Class & Additive \\
            \midrule
            % $\{0, 1\}$      &   $\mathbb{P}$ \citep{barman2018pathtransfers}    \\ 
            % \midrule
            \makecell{$\{a, b\}$ \\ $a = 0, 1$ or $2$}     & \textit{Poly} \citep{akrami2022mnw,akrami2022halfintegers,barman2018pathtransfers}       \\
            \midrule
            \makecell{$\{a, b\}$ \\ $b > a > 3$}     & APX-hard \citep{akrami2022mnw}       \\
            \midrule
            \makecell{$\{a, b\}$ \\ $b > a = 3$}     & \makecell{NP-hard \citep{akrami2022mnw} \\ \textbf{APX-hard} (\Cref{prop:3-c-case})}       \\
            \midrule
            \makecell{$\{0, 1, c\}$ \\ some large $c$}     & APX-hard \citep{garg2018budgetadditive}       \\
            \midrule
            \makecell{$\{a, b, c\}$ \\ $c > b > a \ge 0$}     & \makecell{\textbf{APX-hard} \\ (Theorems \ref{thm:mnw-first-apx-hard} and \ref{thm:mnw-second-apx-hard}) }     \\
            
            \bottomrule
        \end{tabular}
        \caption{Complexity of computing max Nash welfare allocations.}
    \end{subtable}%
    \\
    \begin{subtable}[c]{\linewidth}
        \centering
        \begin{tabular}{ccc}
            \toprule
            Valuation Class & Additive & Submodular\\
            \midrule
            $\{-1, 0\}$       & trivial & \textit{Poly} \citep{barman2023chores}      \\
            \midrule
            $\{-1, 0, 1\}$       & \textit{Poly} \citep{cousins2023mixedmanna} & \textbf{NP-hard} (\Cref{thm:submodular-hard})      \\
            \midrule
            \makecell{$\{-1, 0, c\}$ \\ $c > 1$}      &   \textit{Poly} \citep{cousins2023mixedmanna}  & \textbf{Poly} (\Cref{prop:submodular-orderneutral})     \\
            \midrule
            \makecell{$\{-2, 0, c\}$ \\ $c \ge 3$}       & open & open       \\
            \midrule
            \makecell{$\{a, b, c\}$ \\ all other cases} & \makecell{\textbf{NP-hard} \\ (Theorems \ref{thm:mixedmanna-mew-hard}, \ref{thm:two-negative} and \citep{cousins2023mixedmanna})} & \makecell{\textbf{NP-hard} \\ (Theorems \ref{thm:mixedmanna-mew-hard}, \ref{thm:two-negative} and \citep{cousins2023mixedmanna})} \\
            \bottomrule
        \end{tabular}
        \caption{Complexity of computing max egalitarian welfare allocations.}
    \end{subtable}
    \caption{Summary of our results. Here, $a$, $b$, and $c$ denote arbitrary distinct co-prime integers. Our contributions are highlighted in bold. By valuation class $\{a, b, c\}$, we mean instances where every item is valued at either $a$, $b$, or $c$ by all the agents.}
    \label{table:results}
\end{table*}

The question has been partially answered in the literature. \citet{garg2018budgetadditive} (and \citet{amantidis2021mnw}) show that when each item is valued at $0$, $1$, or $c$, computing a max Nash welfare allocation is APX-hard with a large enough $c$. 
The hardness results for computing max Nash welfare allocations under bivalued valuations \citep{akrami2022mnw,akrami2022halfintegers} also extends to some classes of trivalued valuations. 
Building upon these results, we offer a comprehensive analysis of the complexity of computing fair allocations when agents have ternary or trivalued valuations, i.e., items are valued at arbitrary integers $a$, $b$, or $c$ ($a < b < c$). 
A summary of our results is presented in \Cref{table:results}.

We study the objectives of maximizing Nash welfare and maximizing egalitarian welfare (also known as the Santa Claus objective~\citep{bansal2006santaclaus}). 
These are two of the most popular notions of fairness in the literature, and are extremely well studied. 
The Nash welfare of an allocation is defined as the product of agent utilities, and the egalitarian welfare of an allocation is defined as the utility of the worst-off agent.

We first study the \emph{all goods} setting; here, items have a non-negative marginal value for agents. We show that the problems of computing a max Nash welfare and a max egalitarian welfare allocation are APX-hard for any $a$, $b$, $c$ such that $0 \le a < b < c$ (\Cref{thm:mnw-first-apx-hard,thm:mnw-second-apx-hard}). 
This result completely characterizes the complexity of computing max Nash welfare allocations when agents have ternary valuations. 
Importantly, this result shows that even when agents have $\{0, 1, 2\}$ valuations, computing a max Nash welfare allocation is hard. 
A similar result almost completely characterizes the complexity of maximizing Nash welfare under bivalued valuations \citep{akrami2022mnw,akrami2022halfintegers}; 
these results, however, do not resolve the APX-hardness of the problem for the specific case when one of the values an item can have is $3$.\footnote{For every other case, \citet{akrami2022mnw} show that the problem is either APX-hard or admits an efficient algorithm.} 
We resolve this case as well, showing APX-hardness and completing their characterization under bivalued valuations (\Cref{prop:3-c-case}). 

% Next, we study the \emph{all chores} case; here, all items provide a non-positive marginal value to agents. 
% The complexity of computing egalitarian allocations is resolved by prior work \citep{lenstra1990chores, chakrabarty2015bivalued}. 
% It is known that unless each item is valued at one of two values $a$, $b$ ($a< b$) such that $a = 2b$, the problem is APX-hard. 
% The case where $a = 2b$ admits a polynomial time algorithm. 
% This implies that for all ternary values, unless $a = -2, b = -1$ and $c = 0$, the problem is APX-hard. 
% Note that this special case can be solved similar to the case where $a = 2b$, since any items valued at $0$ by one of the agents can be allocated to them without affecting the egalitarian welfare.

Next, we study the \emph{mixed manna} setting, where items can have both positive and negative marginal values. 
For the special case when $a = -1$, $b = 0$, and $c$ is an arbitrary integer, an efficient algorithm is known to compute max egalitarian welfare allocations \citep{cousins2023mixedmanna}. 
We show that generalizing beyond this case is unlikely by showing that computing a max egalitarian welfare allocation is \np-hard for almost every other $a, b$, and $c$. 

In line with the questions posed by \cite{Babaioff2021Dichotomous} and \cite{cousins2023bivalued},
we also ask the question of whether the results of \citet{cousins2023mixedmanna} can be generalized to submodular valuations. 
We find that apart from the case where $c = 1$, their results can in fact be generalized to submodular valuations (Proposition \ref{prop:submodular-orderneutral}). 
Somewhat surprisingly, for the special case where $a = -1, b= 0$, and $c = 1$, the problem of computing a max egalitarian welfare allocation is \np-hard when agents have submodular valuations (Theorem \ref{thm:submodular-hard}). 

\subsection{Additional Related Work}
Current known results on trivalued valuations consider special cases, e.g., algorithms when $b = 0$ \citep{cousins2023mixedmanna} or hardness when $c$ is much larger than $a$ and $b$ \citep{amantidis2021mnw,garg2018budgetadditive}. 

When all the items are {\em chores}, the complexity of computing egalitarian allocations is resolved by prior work \citep{lenstra1990chores, chakrabarty2015bivalued}. 
It is known that unless each item is valued at one of two values $a$, $b$ ($a< b$) such that $a = 2b$, the problem is APX-hard. 
The case where $a = 2b$ admits a polynomial time algorithm. 
This implies that for all ternary values, unless $a = -2, b = -1$ and $c = 0$, the problem is APX-hard. 
Note that this special case can be solved similar to the case where $a = 2b$, since any items valued at $0$ by one of the agents can be allocated to them without affecting the egalitarian welfare.

Aside from exact algorithms, a long line of fascinating work studies approximation algorithms for maximizing Nash welfare \citep{cole2015nash, Barman2018FindingFA, garg2021rado, li2021mnw,garg2023mnw,dobzinski2024subadditive} and egalitarian welfare \citep{bansal2006santaclaus, annamalai2017santaclaus,chakrabarty2009leximin, lenstra1990chores}. 
The current best known approximation ratios for maximizing Nash welfare are $1.45$ for additive valuations \citep{Barman2018FindingFA} and $4+\epsilon$ for submodular valuations \citep{garg2023mnw}. 
There is also a constant factor algorithm under subadditive valuations which uses a polynomial number of demand queries \citep{dobzinski2024subadditive}. 

There is a $0.5$-approximation algorithm for maximizing egalitarian welfare when all items are chores \citep{lenstra1990chores}, but there is no constant approximation for the all goods case.
However, the special case of the Santa Claus problem admits a $13$-approximation algorithm~\citep{annamalai2017santaclaus}.

\citet{lee:j:apx-hardness} show APX-hardness of the Max Nash welfare problem under general additive valuations, \citet{akrami2022mnw} show APX-hardness for some cases when agents have bivalued additive valuations, and \citet{garg2018budgetadditive} show APX-hardness when agents have $\{0, 1, c\}$ valuations for some large constant $c$. 

\section{Preliminaries}
For any $k \in \mathbb{N}$, we use $[k]$ to denote the set $\{1, 2, \dots, k\}$. 
We have a set of $n$ {\em agents} $N = [n]$ and $m$ {\em items} $O = \{o_1, \dots, o_m\}$. 
Each agent $i \in N$ has a {\em valuation function} $v_i:2^O \rightarrow \R$; $v_i(S)$ specifies agent $i$'s utility for the bundle of items $S$. 
We primarily assume {\em additive} valuations: 
a valuation function is additive if for all $S \subseteq O$, $v_i(S) = \sum_{o \in O} v_i(\{o\})$. 
For readability, we sometimes abuse notation and use $v_i(o)$ to denote $v_i(\{o\})$. 

We often assume that agents have restricted values for items. 
More formally, given a set $A \subseteq \mathbb{Z}$, agent $i$ has $A$-valuations if $v_i$ is additive and $v_i(o) \in A$ for all $o \in O$. 
Specifically, we often consider $\{a, b, c\}$-valuations for some integers $a$, $b$, and $c$. 
Throughout the paper, $a$, $b$, and $c$ will only be used to denote integers.
An {\em allocation} $X = (X_1, \dots, X_n)$ is an $n$-partition of the set of items $O$, where 
agent $i$ receives the bundle $X_i$. 
We require that every item is allocated to some agent. 
This constraint is required when items can be negatively valued by agents. 
The {\em utility} of agent $i$ under the allocation $X$ is $v_i(X_i)$.

\subsection{Fairness Notions (or Objectives)}
We consider two fairness objectives in this paper. 

\noindent\textbf{Max Nash Welfare (MNW):} the Nash welfare of an allocation $X$ is defined as the geometric mean of agent utilities $\NSW(v, X) = \left (\prod_{i \in N} v_i(X_i) \right )^{1/n}$. 
An allocation maximizing this objective value is referred to as the max Nash welfare allocation. 

\noindent\textbf{Max Egalitarian Welfare (MEW):} The egalitarian welfare of an allocation $X$ is defined as the minimum utility obtained by any agent in the allocation i.e. $\min_{i \in N} v_i(X_i)$. An allocation which maximizes this objective value is referred to as the max Egalitarian welfare allocation.

Since the max Nash welfare objective makes little sense when some agents have negative utilities, we only study it when all agents have non-negative utilities. Some of our proofs will also use the utilitarian social welfare of an allocation to establish some bounds on the Nash (or egalitarian) welfare. The {\em utilitarian social welfare} of an allocation $X$ is defined as the sum of agent utilities $\sum_{i \in N} v_i(X_i)$. 

\subsection{Approximation Algorithms and APX-hardness}
For some $\alpha > 1$, an algorithm is an $\alpha$-approximation algorithm for the max Nash welfare objective if it always outputs an allocation which has Nash welfare at least $\frac{1}{\alpha}$ of the optimal value. 
We similarly define an $\alpha$-approximation algorithm for the max egalitarian welfare.

Approximation algorithms are usually only defined when the objective value is either always positive or always negative. This is true of our fairness objectives in the all goods case. 
However, in the mixed goods and chores case, where items can have both positive and negative values, the optimal egalitarian welfare could be positive while many allocations could have a negative egalitarian welfare. 
Thus, in that case, we do not discuss approximability and only discuss \np-hardness and exact algorithms.

For most valuation classes, we show the hardness of computing fair allocations by proving \apx-hardness \citep{papadimitriou1991apx}. 
\apx-hard problems do not admit a Polynomial Time Approximation Scheme (PTAS) unless $\p = \np$. 
This is equivalent to saying that there exists an $\alpha > 1$ such that the problem does not admit an $\alpha$-approximation algorithm unless $\p = \np$.

The class \apx consists of all the problems which admit efficient constant factor approximation schemes. As mentioned in the related work section, existing results put the problem of maximizing Nash welfare in \apx for general additive valuations \citep{cole2015nash,Barman2018FindingFA}. 
These results show that our constant factor lower bounds, are in some sense tight. That is, it would be impossible to improve these lower bounds to $\Omega(\log n)$ or $\Omega(\poly(n))$; the tightest possible lower bounds are constant.

\section{The All Goods Case: \texorpdfstring{$0 \le a < b < c$}{0<=a<b<c}}\label{sec:goods}
We first consider the case where $0 \le a < b < c$ and all agents have $\{a, b, c\}$-valuations. 
It is known that MNW allocations can be computed efficiently when agents have $\{0, 1\}$-valuations \citep{barman2018pathtransfers}, $\{1, c\}$-valuations with $c > 1$ \citep{akrami2022mnw}, and $\{2, c\}$-valuations with $c$ odd and at least $3$ \citep{akrami2022halfintegers}. MEW allocations can be computed efficiently when agents have $\{0, 1\}$-valuations \citep{halpern2020binaryadditive, Babaioff2021Dichotomous}, and $\{1, c\}$-valuations with $c > 1$ \citep{akrami2022mnw, cousins2023bivalued}.

However, the complexity of the $\{0, 1, 2\}$ case is unknown. 
Our first result resolves this.
\begin{restatable}{theorem}{thmmnwfirstapxhard}\label{thm:mnw-first-apx-hard}
When agents have $\{a, b, c\}$-valuations with $0 \le a < b < c$ and $c \le 2b$, computing an MNW allocation is \apx-hard.
\end{restatable}
\begin{proof}
We show \apx-hardness by using the following result from~\cite{ber-kar-sco:j:approx-hardness-max3sat}.

\begin{lemma}[\cite{ber-kar-sco:j:approx-hardness-max3sat}]
Given an instance of 2P2N-3SAT and any $\epsilon > 0$, it is \np-hard to decide if $(1-\epsilon)$
fraction of the clauses can be satisfied or if all solutions satisfy at most a $\frac{1015}{1016}+\epsilon$ fraction
of the clauses.
\end{lemma}
Given an instance $\phi(x_1, \dots, x_n)$ of 2P2N-3SAT with $m = 4n/3$ clauses $C_1, \dots, C_m$, we construct an instance
of an allocation problem with $11n$ items and $8n$ agents.

\medskip
\begin{description}[leftmargin= 0cm]
\item[Items:]
We have the following item types.
\begin{itemize}
    \item For each variable $x_i$ we have five items:
    \begin{itemize}
      \item $x_i$, $x'_i$ corresponding to the positive literals,
      \item $\overline x_i, \overline x'_i$ corresponding to the negative literals,
      \item a clog item $\clog_i$.
    \end{itemize}
   \item $2n$ Type I special items $d_1, \dots, d_{2n}$.
   \item $4n$ Type II special items $\hat{d}_1, \dots, \hat{d}_{4n}$.
 \end{itemize}
\item[Agents:]
We have the following agent types.
 \begin{itemize}
  \item For each variable $x_i$, we have:
  \begin{itemize}
    \item an agent $\Pos_i$ who values $x_i,x'_i$ at $b$ and $\clog_i$ at $c$.
    \item an agent $\Neg_i$ who values $\overline x_i,\overline x'_i$ at $b$ and $\clog_i$ at $c$.
  \end{itemize}
  \item For each clause $C_i$, we have an agent $C_i$ who
  values the items corresponding to the literals in $C_i$ at
  $b$ and the Type I item $d_i$ at $b$.
  \item We have a set of $2n-m$ Type I dummy agents: $s_1, \dots, s_{2n-m}$ where for each $i, 1 \le i \le 2n-m$, $s_i$ values all literal items and the Type I special item $d_{m+i}$ at $b$.
  \item We also have $4n$ Type II dummy agents: $t_1, \dots, t_{4n}$ where
  for each $i, 1 \le i \le 4n$, $t_i$ values the Type II special item $\hat{d}_i$ at $c$, and $t_i$ values all literal items at $b$.
  \end{itemize}
\end{description}
All unmentioned valuations are at $a$. 
Note that each of the Type I (resp. Type II) items are
valued by exactly one agent at $b$ (resp. $c$).

We now prove correctness of our reduction. We have two cases to handle in our proof.
\begin{description}[leftmargin=0cm]
    \item[Case 1: $b^2 \ge ca$:]
    ($\Longrightarrow$) Suppose there is an assignment $\sigma:\{1, \dots,n\} \rightarrow \{0, 1\}$ to the variables $x_1, \dots, x_n$ that satisfies at least $(1-\epsilon)$ of the clauses in $\phi$. 
    We construct an allocation % ${\cal X}$ 
    using this assignment.

We first allocate the items to the variable agents.
For each $i, 1 \le i \le n$, if $\sigma(x_i) = 1$ we allocate $\clog_i$ to $\Pos_i$ and $\overline x_i,\overline x'_i$ to $\Neg_i$; if $\sigma(x_i) = 0$ we allocate $\clog_i$ to $\Neg_i$ and ${x_i},{x'_i}$ to $\Pos_i$.
Thus, if $\sigma(x_i) = 1$, the agent $\Pos_i$ gets the clog item for a utility of $c$, and the agent $\Neg_i$ gets a utility of $2b$ from the two literal items assigned to it. 
%%%%
%% Variable agents get c or 2b
%%%%

%Each other agent has an item they uniquely value at $b$ (or $c$) and we simply allocate these Type I (Type II) items to the corresponding agents. 

For each clause $C_i$ that is satisfied by the assignment $\sigma$, we allocate exactly one copy of one of the literal items that satisfies that clause to the corresponding clause agent. 
For example, if $C_i = (\overline x_1 \vee x_2 \vee x_3)$ and $\sigma(x_1) = 0, \sigma(x_2) = 0,$ and $\sigma(x_3) = 1$ then we can allocate $\overline x_1$, $\overline x'_1$, $x_3$ or $x'_3$ to the clause agent $C_i$.
Thus, if $C_i$ is a satisfied clause, we allocate one literal item to its corresponding agent for a utility of $b$.
Each clause agent uniquely values a Type I special item at $b$, and we allocate this item to them. Overall, each clause agent corresponding to a satisfied clause receives a utility of $2b$ and each clause agent corresponding to an unsatisfied clause receives a utility of $b$.
%%%%
%% Satisfied clauses get 2b, unsatisfied clauses get b.
%%%%

So far, we have allocated $2n$ literal items to the variable agents, $\le m$ literal items to the satisfied clause agents, and $m$ of the Type I special items to their corresponding clause agents.
For each of the Type I dummy agents, we allocate one of the remaining literal items. 
There are at least $2n -m$ such literal items remaining so this is possible. 
Once this is done, notice that we have at most $\epsilon m$ literal items remaining since at least $(1-\epsilon)m$ clauses are satisfied by the assignment $\sigma$. 
To complete our allocation of the literal items, we iteratively allocate these $\epsilon m$ literal items 
to Type II dummy agents such that each Type II dummy agent does not receive more than one such item.
We have $2n-m$ remaining Type I items and $4n$ remaining Type II items. We allocate each remaining Type I item to its corresponding Type I dummy agent that values it at $b$, and each Type II item to its corresponding Type II dummy agent that values it at $c$.
%%%%
%% Type I get 2b 
%% m' of Type II get c+b
%% 4n-m' Type II get c
%%%%
Let us take stock of the utility of our agents at this stage. There are:
\begin{itemize}
    \item $n$ literal agents who received two literal items for a utility of $2b$.
    \item $n$ literal agents who received the clog item for a utility of $c$.
    \item $m - m'$ satisfied clause agents who receive one literal item for a utility of $b$ and their corresponding Type I item for a utility of $2b$.
    \item $m'$ unsatisfied clause agents who receive their corresponding Type I item (and no literal item) for a utility of $b$. 
    \item $2n- m$ Type I dummy agents who receive a literal item which they value at $b$, and their corresponding Type I item which they value at $b$, for a total utility of $2b$.
    \item $m'$ Type II dummy agents who receive a literal item and their corresponding Type II item which they value at $c$, for a total utility of $c+b$.
    \item $4n-m'$ Type II dummy agents who receive only their
    corresponding Type I item for a total utility of $c$.
\end{itemize}
Above, we use $m'$ to denote the number of unsatisfied clauses. We have the following NSW for our allocation.
 
\begin{align}
\NSW(v, X) = \left (\prod_{i \in N} v_i({X}_i) \right )^{\frac{1}{8n}} &= ((2b)^n c^n (2b)^{m-m'} b^{m'} (2b)^{2n-m} (c+b)^{m'} c^{4n-m'})^{\frac{1}{8n}} \notag\\
&= \left ((2b)^{3n} c^{5n} \left (\frac{b(c+b)}{2b \cdot c}\right )^{m'} \right )^{\frac{1}{8n}}\notag\\
&\ge \left ((2b)^{3n} c^{5n} \left (\frac{b(c+b)}{2b \cdot c}\right )^{\epsilon m} \right )^{\frac{1}{8n}}\label{eq:NSW-abc-goods-value}
\end{align}
The final inequality holds since $b(b+c) < 2b \cdot c$ and $m' \le \epsilon m$.

($\Longleftarrow$) For the other direction, consider an arbitrary Max Nash Welfare allocation $X$. 
We state some important properties of $X$ and upper-bound its Nash Welfare. 
These properties can be assumed without loss of generality; that is, these properties are satisfied by at least one max Nash welfare allocation $X$. 
At a high level, these properties show that $X$ should not allocate a value of $a$ to any agent, and subject to this constraint, must be as egalitarian as possible.  
\begin{property}\label{prop:all-positive}
If $X$ maximizes the NSW, then all agents receive at least one item that gives them
a positive utility.
\end{property}
\begin{proof}
If the allocation does not do this the Nash welfare is 0. However, it is easy to find an allocation with positive Nash welfare. The non-literal agents can get their special item and the literal agents can share their clog and literal items so that they get a positive utility.
\end{proof}

\begin{property}\label{prop:clog-assignment}
If $X$ maximizes the NSW, the item $\clog_i$ is allocated either to
$\Pos_i$ or $\Neg_i$.
\end{property}
\begin{proof}
Suppose that $\Pos_i$ did not receive $\clog_i$.
From Property \ref{prop:all-positive} we know that $\Pos_i$ is assigned at least one item that gives them a positive utility. 
If $\clog_i$ is allocated to a non-literal agent, they value the item at $a$ and so we can swap
$\clog_i$ with an item allocated to $\Pos_i$. 
This swap weakly improves the Nash welfare of the allocation.
\end{proof}

The same argument as above can be used to show the following property
concerning the special items.

\begin{property}\label{prop:special-items}
If $X$ maximizes the NSW, a special item is never allocated to an agent who
values it at $a$.
\end{property}

At this point, we have fixed the allocation of all the $6n$ special items and the $n$ clog items in $X$: special items go to agents who value them at $b$ or $c$, and clog items go to the literal agents. 
This leaves us only with the $4n$ literal items. 
Since there are $4n$ Type II dummy agents, if any non-Type II dummy agent has more than one item, there must be at least one Type II dummy agent who does not receive a literal item, and only receives their special item. 
We will use this property extensively in our analysis. 

\begin{property}\label{prop:literal-assignment}
If $X$ maximizes the NSW then for each $i \in [n]$,
\begin{enumerate}[(a)]
    \item if $\clog_i \in X_{\Pos_i}$, then $|X_{\Pos_i}| = 1$ and
$|X_{\Neg_i}| \le 2$.
    \item if $\clog_i \in X_{\Neg_i}$, then $|X_{\Neg_i}| = 1$ and
$|X_{\Pos_i}| \le 2$.
\end{enumerate}
\end{property}
\begin{proof}
We only show (a), as an analogous argument holds for (b). 
We assume that the literal agent $\Pos_i$ has the clog item $\clog_i$. If $\Pos_i$'s bundle has any other item, then some Type II dummy agent has a single item. We can move any item $o \in X_{\Pos_i} \setminus \{\clog_i\}$ to a Type II dummy agent $t_j$ such that $|X_{t_j}| = 1$ (i.e.,
$t_j$ is allocated only its special item).
This results in a weak improvement of Nash welfare since
\begin{align*}
    \frac{v_{\Pos_i}(X_{\Pos_i}) - v_{\Pos_i}(o)}{v_{\Pos_i}(X_{\Pos_i})} \ge \frac{c}{c+b} = \frac{v_{t_j}(X_{t_j})}{v_{t_j}(X_{t_j}) + v_{t_j}(o)},
\end{align*}
which implies 
\begin{align*}
    (v_{\Pos_i}(X_{\Pos_i}) - v_{\Pos_i}(o))(v_{t_j}(X_{t_j}) + v_{t_j}(o)) \ge v_{\Pos_i}(X_{\Pos_i})v_{t_j}(X_{t_j})
\end{align*}
To show the second part, assume $|X_{\Neg_i}| \ge 3$. 
By a similar argument, moving the least-valued item $o$ allocated to $\Neg_i$ to a Type II dummy agent $t_j$ who was only given one item weakly increases Nash welfare. This follows since
    \begin{align*}
        \frac{v_{neg_i}(X_{neg_i}) - v_{neg_i}(o)}{v_{neg_i}(X_{neg_i})} \ge \frac{2}{3} \ge \frac{b}{c+b} = \frac{v_{t_j}(X_{t_j})}{v_{t_j}(X_{t_j}) + v_{t_j}(o)},
    \end{align*}
    which implies
    \begin{align*}
    (v_{neg_i}(X_{neg_i}) - v_{neg_i}(o))(v_{t_j}(X_{t_j}) + v_{t_j}(o)) \ge v_{neg_i}(X_{neg_i})v_{t_j}(X_{t_j}).
\end{align*}
\end{proof}
We can use the second part of the proof of Property \ref{prop:literal-assignment} to show the following
\begin{property}\label{prop:bundle-size}
In any NSW maximizing allocation ${X}$, no agent has a bundle of size at least 3.
\end{property}
Properties \ref{prop:literal-assignment} and \ref{prop:bundle-size} use
the assumption that $2b \ge c$. 
We have not yet used the case assumption that $b^2 \ge ca$, which we do next.

\begin{property}\label{prop:utlity-b-vs-a}
If ${X}$ maximizes the NSW, then if any agent has an item that gives them a utility of
$b$, they do not have an item that gives them a utility of $a$.
\end{property}
\begin{proof}
Assume an agent $i$ has at least two items, and one of these items $o$ gives them a utility of $a$. 
Note that $i$ cannot be a dummy agent because of Property \ref{prop:special-items}: special items never go to agents that value them at $a$; thus, the item $o$ must be a literal item which offers them a utility of $b$. 
Since $i$ has at least two items in $X$, there must be a Type II dummy agent $t_j$ who receives exactly one item in $X$. 
We move the item $o$ to $X_{t_j}$. This weakly improves Nash welfare since: 
\begin{align*}
    \frac{v_{i}(X_{i}) - v_{i}(o)}{v_{i}(X_{i})} \ge \frac{b}{a+b} \ge \frac{c}{c+b} = \frac{v_{t_j}(X_{t_j})}{v_{t_j}(X_{t_j}) + v_{t_j}(o)}.
\end{align*}
The second inequality holds since $a \le \frac{b^2}{c}$.
This implies 
\begin{align*}
    (v_{i}(X_{i}) - v_{i}(o))(v_{t_j}(X_{t_j}) + v_{t_j}(o)) \ge v_{i}(X_{i})v_{t_j}(X_{t_j}).
\end{align*}
\end{proof}

\begin{property}\label{prop:typei-exactly-two}
If ${X}$ maximizes the NSW, all Type I dummy agents must receive exactly two items.
\end{property}
\begin{proof}
If there is a Type I dummy agent $s_i$ who receives exactly one item, then using Properties \ref{prop:literal-assignment} and \ref{prop:bundle-size}, there is at least one Type II dummy agent $t_j$ who receives a bundle of size $2$. This bundle $X_{t_j}$ must contain one literal item $o$. Moving $o$ to $X_{s_i}$ strictly improves Nash welfare.
\end{proof}
Note the key difference in language in the above lemma. We use the word `must' because if the above property is not satisfied, the allocation $X$ is not max Nash welfare. This stronger property will be used to prove the next property.

For the next property, we derive a truth assignment $\sigma$ for the 2P2N-3SAT instance $\phi$ from the allocation ${X}$.
If $\clog_i$ is allocated to $\Pos_i$ then
$\sigma(x_i) = 1$ otherwise $\sigma(x_i) = 0$.
Recall that all assignments to $\phi$ satisfy at most $(\frac{1015}{1016}+\epsilon$) fraction of the clauses, and so $\sigma$ satisfies some, but not all, of the clauses.

\begin{property}\label{prop:satisfied-clause}
In ${X}$, any clause $C_i$ that is satisfied by the
assignment $\sigma$ receives exactly one item corresponding to a copy
of a literal which satisfies it.
\end{property}
\begin{proof}
We know from Properties~\ref{prop:all-positive} and~\ref{prop:bundle-size} that $C_i$ is allocated either
one or two items.

If $|X_{C_i}| = 2$ and the property does not hold,
then $X_{C_i}$ must have a copy of a literal item $o$ which
satisfies it but is not part of the assignment (e.g., $C_i$
receives $\overline x_j$, but $\sigma(x_j) = 1$). 
This follows from Property \ref{prop:utlity-b-vs-a}.

There must be at least one literal (and so two literal items)
which satisfy the clause, and only one other clause which can be
satisfied by this literal. Therefore one of the items
corresponding to a literal which satisfies the clause is allocated to one of the dummy agents. 
Swapping the item $o$ with the dummy agent achieves the required lemma while keeping
Nash Welfare unchanged.

If $|X_{C_i}| = 1$, then by the same argument we know that one
of the items corresponding to a literal which satisfies the clause is allocated to one of the dummy agents. 
We can move this item to $X_{C_i}$ which will weakly improve Nash Welfare. 
Note that this swap may violate Property \ref{prop:typei-exactly-two} but if that happens, then we can conclude using the stronger statement of Property \ref{prop:typei-exactly-two} that $X$ was not a max Nash welfare allocation to begin with. 
\end{proof}

\begin{property}\label{prop:unsatisfied-clause}
In ${X}$, for any clause $C_i$ not satisfied by the
assignment $\sigma$, the corresponding clause agent receives exactly 
one item.
\end{property}

\begin{proof}
Assume that $C_i$ is not satisfied by the assignment $\sigma$ and that $|X_{C_i}| = 2$. 
By Property~\ref{prop:utlity-b-vs-a} the clause agent is given no items it values at $a$, and thus, its items provide a value of $b$. 
Thus, $C_i$'s bundle contains the clause agent's unique dummy item, and one literal item $x_j$.
Since this clause is not satisfied by the assignment $\sigma$, $X_{\Pos_j}$ must have either one item or two items where both have value $a$.

If the first case holds, we simply move $x_j$ to $X_{\Pos_j}$ resulting
in a weak improvement in Nash welfare. 
This transfer could potentially result in agent $\Pos_j$ violating Property \ref{prop:utlity-b-vs-a}. 
If this happens, transferring items according to the proof of Property \ref{prop:utlity-b-vs-a} resolves this without violating this Property.

If the second case holds, we swap one of the items in
$X_{\Pos_j}$ with $x_j$ resulting in a strict improvement in
Nash Welfare (since $(b+a)\cdot (b+a) > 2a \cdot 2b$), contradicting the fact that $X$ is a max Nash welfare allocation.
\end{proof}

\begin{property}\label{prop:clog-literal-encoding}
Assume $b^2 \ge ac$. If ${X}$ maximizes the NSW then for each $i \in [n]$,
\begin{enumerate}[(a)]
    \item if $\clog_i \in {X}_{\Pos_i}$, then
$|{X}_{\Neg_i}| = 2$ and $v_{\Neg_i}({X}_{{\Neg}_i}) = 2b$, and
    \item if $\clog_i \in {X}_{\Neg_i}$, then
$|{X}_{\Pos_i}| = 2$ and $v_{\Pos_i}({X}_{\Pos_i}) = 2b$.
\end{enumerate}
\end{property}
\begin{proof}
We prove part (a) here; part (b) follows similarly. 
If the Property does not hold, from Properties \ref{prop:bundle-size} and \ref{prop:utlity-b-vs-a}, $\Neg_i$ must have either one item or two items where both items are valued at $a$.

From Properties \ref{prop:unsatisfied-clause} and \ref{prop:utlity-b-vs-a}, the literal items $\overline x_i$ and $\overline x'_i$ that do not belong to $\Neg_i$ must belong to dummy agents.

If $\Neg_i$ has two items where both items are valued at $a$, we swap one of these items with $\overline x_i$ and strictly improve on the Nash welfare (since $(b+a)\cdot 2b > 2a \cdot 2b$). 
This contradicts the fact that $X$ is a max Nash welfare allocation and shows this case cannot occur.

If $\Neg_i$ receives one item and the item has utility $a$, moving $\overline x_i$ to $X_{\Neg_i}$ strictly improves Nash welfare again (since $(b+a)\cdot b > a \cdot 2b$), so this cannot happen as well. 

The only case left is when $\Neg_i$ receives one item and the item has utility $b$. 
Assume without loss of generality that this item is $\overline x'_i$. We move $\overline x_i$ to $\Neg_i$ and weakly improve the Nash welfare. 
If $\overline x_i$ originally belonged to a Type I dummy agent, this contradicts Property \ref{prop:typei-exactly-two} which means that $X$ was not max Nash welfare to begin with. 
Otherwise, the transfer weakly improves Nash welfare without violating any of the other properties.
\end{proof}
Crucially, it can be seen that there exists a Nash welfare maximizing allocation $X$ that satisfies all ten properties. 
This is because we ensure all the transfers outlined in the proofs do not violate any other property. 
This allows us to use all ten properties to find the max Nash welfare in our instance.
\begin{enumerate}[(i)]
    \item We first handle the literal agents. For each $i \in [n]$, $v_{\Pos_i}({X}_{\Pos_i})v_{\Neg_i}({X}_{\Neg_i}) = c \cdot 2b$
    by Properties~\ref{prop:literal-assignment} and~\ref{prop:clog-literal-encoding}.
    \item For each Type I dummy agent $s_i$, $v_{s_i}({X}_{s_i}) = 2b$ by Properties~\ref{prop:special-items}, \ref{prop:utlity-b-vs-a}, and~\ref{prop:typei-exactly-two}.
    \item For each clause agent $C_i$, if it is satisfied by the assignment
    $\sigma$, $v_{C_i}({X}_{C_i}) = 2b$ by Properties~\ref{prop:utlity-b-vs-a} and~\ref{prop:satisfied-clause}, and otherwise $v_{C_i}({X}_{C_i}) = b$ by Properties~\ref{prop:special-items} and~\ref{prop:unsatisfied-clause}.
    \item Let $m'$ be the number of unsatisfied clauses using the
    assignment $\sigma$. From the above statements we
    can see that the utility of all other agents is set and $m'+4n$ items remain. It is straightforward to see that the best way to allocate these items is to allocate a utility of $b+c$ to $m'$ of the Type II dummy agents with the remaining having utility $c$. This handles allocating
    all of the remaining items.
\end{enumerate}

We can now calculate the Nash welfare of our allocation ${X}$.

\begin{align*}
\NSW(v, X) = \left (\prod_{i \in N} v_i({X}_i)\right )^{\frac{1}{8n}} &= (c^{n} (2b)^{n} (2b)^{2n-m} (2b)^{m-m'} b^{m'} (b+c)^{m'}c^{4n-m'})^{\frac{1}{8n}}\\
% &= (c^{5n-m'}(2b)^{3n-m'}b^{m'}(b+c)^{m'})^{\frac{1}{8n}}\\
&= \left(c^{5n} (2b)^{3n} \left(\frac{b(c+b)}{c(2b)}\right)^{m'}\right)^{\frac{1}{8n}}\\
&\le \left(c^{5n} (2b)^{3n} \left(\frac{b(c+b)}{c(2b)}\right)^{\left(\frac{1}{1016}-\epsilon \right)m}\right)^{\frac{1}{8n}}
\end{align*}
Our last inequality follows from $m' \ge \left (\frac{1}{1016} - \epsilon \right )m$.

We can calculate our final approximation lower bound by dividing the two
Nash welfares to yield the following lower bound.

\[
\left (\left (\frac{b(c+b)}{c(2b)}\right)^{\left (\frac{-1}{1016}+2\epsilon \right )m}\right)^{\frac{1}{8n}}
\]

Note that since our allocation instance was constructed from a boolean
formula in 3CNF where each variable occurs twice as a positive literal and
twice as a negative literal, $m = 4n/3$. Thus we can restate our lower bound
as the following constant value.

\[
\left (\frac{2bc}{b(c+b)}\right)^{\left(\frac{1}{1016}-2\epsilon\right)\frac{1}{6}}
\]

The above statement evaluates to $1.00004$
for $\{0,1,2\}$ valuations with
a small enough $\epsilon$.
\item[Case 2: $b^2 < ca$:]
The main idea in this case is that $a$ is considered a high value; so, instead of adding items to the dummy
agents, we add them to the clause agents corresponding to unsatisfied clauses who value them at $a$. 
% In this case we can see that $a$ is considered
% a high value.

($\Longrightarrow$) Suppose there is an assignment $\sigma$
which satisfies $(1-\epsilon)m$ clauses.
We follow the same allocation as in the other case except that we now give
the ``extra'' $m'$ items to the agents corresponding to unsatisfied clauses.
This gives these agents a utility of $b+a$ and the following overall Nash welfare.

\begin{align*}
    \NSW(v, X) &= (c^{5n}(2b)^{3n-m'}(b+a)^{m'})^{\frac{1}{8n}} = \left (c^{5n}(2b)^{3n} \left (\frac{b+a}{2b} \right )^{m'} \right )^{\frac{1}{8n}} \\
    &\ge \left (c^{5n}(2b)^{3n} \left (\frac{b+a}{2b} \right )^{\epsilon m} \right )^{\frac{1}{8n}}
\end{align*}
The last inequality follows from $m' \le \epsilon m$.

($\Longleftarrow$) For the other direction, we again suppose that all assignments to
$\phi$ satisfy at most $(\frac{1015}{1016} + \epsilon)m$
clauses and consider an arbitrary max Nash welfare allocation ${X}$.

First notice that Properties~\ref{prop:all-positive}, \ref{prop:clog-assignment}, \ref{prop:special-items}, \ref{prop:literal-assignment}, and~\ref{prop:bundle-size} hold in this case as well. We will have an
alternative property to Property~\ref{prop:utlity-b-vs-a}, since there will
be cases here where agents will be allocated items they value at $a$.

\begin{property}\label{prop:other-typeii-dummy}
In ${X}$, all Type II dummy agents $t_i$ must have $|{X}_{t_i}| = 1$.
\end{property}
\begin{proof}
Suppose that $|X_{t_i}| > 1$ for some Type II dummy agent $t_i$. 
Then there must be an agent who does not receive $\clog_k$ for some $k$ and is not a type II dummy agent, who receives a bundle of size $1$. 
This comes from the fact that there are $3n$ agents who are neither type II dummy agents nor do they receive $\clog_k$ for some $k$, and there are $3n$ more items than agents. 
Let this agent be $j$. Moving a literal item $o$ from $X_{t_i}$ to $X_j$ strictly improves Nash welfare since:
\begin{align*}
    \frac{v_{t_i}(X_{t_i}) - v_{t_i}(o)}{v_{t_i}(X_{t_i})} \ge \frac{c}{b+c} > \frac{b}{a+b} \ge \frac{v_j(X_j)}{v_j(X_j) + v_j(o)}
\end{align*}
The strict inequality follows from $c > \frac{b^2}{a}$.
%
% It is straightforward to see that moving an item to another agent who receives only one item with value less than $c$ would strictly increase Nash welfare. This is a contradiction.
\end{proof}

Note again the stronger language used in the statement. 
If the above property is not satisfied, then $X$ is not a max Nash welfare allocation. 
It is also worth noting that the proofs of Properties \ref{prop:literal-assignment} and \ref{prop:bundle-size} use transfers that violate this property but weakly improve Nash welfare. 
This only shows that the Properties \ref{prop:literal-assignment} and \ref{prop:bundle-size} must be satisfied in any max Nash welfare allocation $X$, much like Property \ref{prop:other-typeii-dummy}.

Using Properties  \ref{prop:literal-assignment}, \ref{prop:bundle-size}, and \ref{prop:other-typeii-dummy}, we get that any agent who either receives $\clog_i$ or is a type II dummy agent, must receive exactly one item. 
This leaves us with $3n$ agents and $6n$ items with agent bundles upper bounded at size two (Property \ref{prop:bundle-size}). 
This means that all these $3n$ agents must receive exactly two items. 

Finally, we show an analog to Property \ref{prop:clog-literal-encoding}, showing that the variable agent that does not receive a clog item receives two items of value $b$ each.

% The remaining properties from the other case hold with similar arguments as
% used in the case where $b^2 \ge ca$. The main idea behind new arguments is
% that agents with $c$-valued items have exactly one item and all other agents have exactly two items.
% \zfnote{We could further explain how these properties hold as stated in your
% notes, but this may read better.}

\begin{property}\label{prop:other-clog-literal-encoding}
Assume $b^2 < ac$. In any NSW maximizing allocation ${X}$, for each $i \in [n]$,
\begin{enumerate}[(a)]
    \item if $\clog_i \in {X}_{\Pos_i}$, then
$|{X}_{\Neg_i}| = 2$ and $v_{\Neg_i}({X}_{\Neg_i}) = 2b$, and
    \item if $\clog_i \in {X}_{\Neg_i}$, then
$|{X}_{\Pos_i}| = 2$ and $v_{\Pos_i}({X}_{\Pos_i}) = 2b$.
\end{enumerate}
\end{property}
\begin{proof}
We prove part (a) here; part (b) follows similarly. In this case $\Neg_i$ must have two items that it values at $b$. 
If the property is not satisfied, at least one of these items must provide a utility of $a$. 
The items $\overline x_i$ and $\overline x'_i$ that do not belong to $\Neg_i$ must belong to other clause or Type I dummy agents.

If both items in $X_{\Neg_i}$ provide a utility of $a$, then we swap one of these items with $\overline x_i$ or $\overline x_i'$. 
It is easy to see that this transfer strictly improves Nash welfare, and so this case cannot happen.

Assume one of the items $o$ in $X_{\Neg_i}$ provides $\Neg_i$ a utility of $a$. 
Without loss of generality, assume $\overline x_i$ belongs to some other agent. 
It is easy to see that swapping $\overline x_i$ with $o$ weakly improves Nash welfare.
\end{proof}

Consider the assignment $\sigma$ such that if $\clog_i \in X_{\Pos_i}$, $\sigma(x_i) = 1$; and $\sigma(x_i) = 0$ otherwise. 
The above property implies that if any clause which is not satisfied by the assignment $\sigma$ receives a utility of $b+a$ since all the literal items the clause values at $b$ is allocated to the corresponding literal agents.

These properties allow us to upper bound the max Nash welfare, very similarly to the previous case.
\begin{align*}
\NSW(v, X) = \left(\prod_{i \in N} v_i({X}_i)\right)^{\frac{1}{8n}} &= \left (c^{5n}(2b)^{3n-m'}(b+a)^{m'} \right )^{\frac{1}{8n}}\\
 &= \left (c^{5n}(2b)^{3n}\left (\frac{b+a}{2b} \right )^{m'}\right )^{\frac{1}{8n}}\\
 &\le \left(c^{5n}(2b)^{3n}\left (\frac{b+a}{2b} \right )^{\left (\frac{1}{1016}-\epsilon \right )m}\right)^{\frac{1}{8n}}
\end{align*}
We can now calculate the lower bound by taking the ratio of our calculated Nash welfares.
\begin{align*}
\left(\frac{2b}{b+a}\right)^{\left (\frac{1}{1016}-2\epsilon \right )m\frac{1}{8n}} = \left(\frac{2b}{b+a}\right)^{\frac16\left (\frac{1}{1016}-2\epsilon \right)}
\end{align*}
Notice that this is a constant. This completes our proof.
\end{description}
\end{proof}

Recently, \citet{jain2024matching} studied the problem of maximizing Nash welfare under two sided preferences and show that the problem is \np-hard under $\{0, 1, 2\}$ valuations and capacity constraints. Since fair allocation is a special case of their problem, Theorem \ref{thm:mnw-first-apx-hard} presents an improved hardness result for their problem since we show \apx-hardness and eliminate the need for capacity constraints. 

Our next result resolves the case when $2b < c$. It may be possible to use a similar 2P2N-3SAT construction in this case as well but our proof uses a much simpler vertex cover based reduction. Specifically, we reduce from the following problem.

\begin{restatable}{lemma}{lemmaxthreevchardness}\label{lem:max-3-vc-hardness}
There exists a constant $\gamma \in (0, 1)$ such that, given a $3$-regular graph $G$ and an integer $k$, it is \np-hard to decide if $G$ has a vertex cover of size $k$ that covers all edges or all subsets of nodes of size $k$ cover at most a $(1-\gamma)$ fraction of the edges of $G$.   
\end{restatable}

This follows from applying the arguments of \citet{pet:j:max-vertex} to the min vertex cover hardness result of \citet{chlebik2006minvc}. The proof can be found in Appendix \ref{apdx:goods}.

\begin{theorem}\label{thm:mnw-second-apx-hard}
When agents have $\{a, b, c\}$-valuations with $0 \le a < b < c$ and $2b < c$, computing an MNW allocation is \apx-hard.
\end{theorem}

\begin{proof}

We reduce from Lemma \ref{lem:max-3-vc-hardness}. We are given a $3$-regular graph $G = (V, E)$ and an integer $k$. Note that since the graph is $3$-regular, $|E| = \frac{3|V|}{2}$, and $k$ must be at least $\frac{|V|}{2}$ for this problem to be non-trivial. This is because we need at least $\frac{|V|}{2}$ nodes to cover $\frac{3|V|}{2}$ edges in a $3$-regular graph.

We construct a fair allocation instance with $3k - 0.5|V|$ agents and $7k - 1.5|V|$ items.

The $7k - 1.5|V|$ items are defined as follows:
\begin{enumerate}[(a)]
    \item For each each edge $e_{ij} \in E$ we have an item $e_{ij}$,
    \item We have $k$ {\em vertex cover} items $c_1, \dots, c_k$, and 
    \item We have $6k - 3|V|$ special items.
\end{enumerate}

The $3k - 0.5|V|$ agents have valuations defined as follows:
\begin{description}[leftmargin=0cm]
    \item[Node Agents:] For each node $i \in V$, we have an agent who values the edges incident on it at $b$ and
    \item[Dummies:] We have $3k - 1.5|V|$ dummy agents, who value edge items at $b$, and exactly two special items each at $b$. Since there are $6k-3|V|$ special items, we can ensure that no two dummy agents value the same special item at $b$.
\end{description}
The vertex cover items are valued at $c$ by both the node and dummy agents. 
All unmentioned values are at $a$.

We now prove correctness of our reduction.

Assume the graph admits a vertex cover (say $S$) of size $k$. We allocate the $k$ vertex cover items to the agents in $S$. All other node agents receive the edge items corresponding to the three edges they are incident on. 
This is feasible since at least one endpoint of each edge belongs to the vertex cover $S$.

At this point, we only have to allocate the special items, and perhaps some edge items (if both endpoints of some edge belong to the vertex cover $S$).
We allocate the remaining items to the dummy agents. Each dummy agent receives their two special items and a single unassigned edge item; this ensures all $7k - 1.5|V|$ items are assigned.

The $k$ agents in the vertex cover $S$ have a utility of $c$, and the remaining $|V|-k$ node agents have a utility of $3b$. The $3k-1.5|V|$ dummy agents have a utility of $3b$ as well. 
Thus, the Nash welfare of this allocation is 
\begin{align*}
    \NSW(v, X) = \left (c^k (3b)^{(2k - 0.5|V|)} \right )^{\frac{1}{3k - 0.5|V|}} && (\NSW^+)
\end{align*}

For the other direction, assume that no subset of nodes of size $k$ covers more than $(1-\gamma)|E|$ edges. 
For this case, we slightly tweak the valuation function to make our analysis easier. Specifically, we change the valuation functions such that all valuations at $a$ are replaced with $a'$ such that $a' = \max \{a, \frac{2b^2}{c}, \frac{2b}{3}\}$. $a'$ may no longer be an integer but it is guaranteed to be less than $b$ since $2b < c$. We refer to this new valuation profile using $v'$. Crucially, since we only increased agent valuations, we must have, for any allocation $X$, $\NSW(v, X) \le \NSW(v', X)$. 
Thus, any upper bound on the max Nash welfare under the valuations $v'$ also bounds the max Nash welfare under the valuations $v$. 
To upper bound the Nash welfare, we examine the MNW allocation $X$.

\begin{lemma}\label{lem:vc-properties}
There exists an MNW allocation $X$ with the
following properties.
\begin{enumerate}
\item No agent receives more than one vertex cover item.
\item An agent who receives a vertex cover item does not receive any other item in $X$.
\item No agent receives four or more items.
\end{enumerate}
\end{lemma}

\begin{proof}
We separately consider each property of $X$ stated in the lemma.

\begin{enumerate}
\item If an agent has two vertex cover items, there must be some agent $j$ who does not receive a vertex cover item and receives at most two items. 
Moving one of the vertex cover items to agent $j$'s bundle strictly improves Nash welfare. This follows from the fact that $c > 2b \ge v'_j(X_j)$.

\item If an agent $i$ with a vertex cover item $c_r$
has another item (say $o$), then there must be an agent $j$ without a vertex cover item that has at most two items. Transferring $o$ to agent $j$ weakly improves the Nash welfare. This is because 
\begin{align*}
    \frac{v'_i(X_i) - v'_i(o)}{v'_i(X_i)} \ge \frac{c}{b+c} \ge \frac{2b}{a' + 2b} \ge \frac{v'_j(X_j)}{v'_j(X_j) + v'_j(o)}
\end{align*}
The second inequality follows by plugging in $c \ge \frac{2b^2}{a'}$.

\item If an agent has four or more items, we can transfer the least valued item (out of the four or more) to an agent who receives at most two items and no vertex cover items. Crucially, this uses the fact that $a' \ge \frac{2b}{3}$.
\end{enumerate}
Note that the transfers to show properties 2 and 3 only weakly improve Nash welfare, but these transfers are made without violating the other properties; additionally, these transfers need to be made only a finite number of times to ensure the property holds. So there must exist an MNW allocation where all three properties are satisfied.
\end{proof}

The three properties stated in Lemma~\ref{lem:vc-properties}
offer us some structure. Each agent in the max Nash welfare allocation $X$ either has a vertex cover item or exactly three other items. 

Next, we lower bound the number of agents who receive three items but do not receive a utility of $3b$. 
Consider the subset of nodes consisting of the agents who receive a vertex cover item. 
This subset of nodes, by assumption, must \emph{not} cover at least $\gamma|E|$ edges. 
Each of the uncovered edges represents an edge item that both endpoints value at $b$, but can only be given to one of them. 
Thus, one of the uncovered edge's endpoints must receive a utility of less than $3b$. 
If we do this for all uncovered edges, we get that there are at least $\frac{\gamma|E|}{3}$ agents whose utility is less than $3b$. 
We divide by three since $G$ is $3$-regular, so each node is counted at most thrice.

These $\frac{\gamma|E|}{3}$ agents receive a utility of at most $2b + a'$. All other agents receive a utility of either $3b$ or $c$. Note that receiving a utility of more than $3b$ is impossible with only three items, as the only items valued at $c$ are the vertex cover items.
This upper bounds the Nash welfare of the allocation $X$ under $v'$. 
Assuming $m'$ agents do not receive a vertex cover item or a utility of $3b$, $\NSW(v', X)$ is upper bounded by
\begin{align*}
    &\left (c^k (3b)^{2k - 0.5|V| - m'} (2b + a')^{m'} \right )^{\frac{1}{3k - 0.5|V|}} \\ &\quad \le
    \left (c^k (3b)^{2k - 0.5|V|} \left (\frac{2b + a'}{3b}\right )^{\frac{\gamma |E|}{3}} \right )^{\frac{1}{3k - 0.5|V|}} &\quad (\NSW^{-})
\end{align*}
The inequality holds since $m' \ge \frac{\gamma |E|}{3}$.
We have shown that it is \np-hard to decide whether an allocation has Nash welfare at least $\NSW^+$ or whether all allocations have a Nash welfare of at most $\NSW^{-}$. 
Taking the ratio of the two gives us the following approximation lower bound
\begin{align*}
    \frac{\NSW^+}{\NSW^{-}} = \left (\left (\frac{3b}{2b + a'}\right )^{\frac{\gamma |E|}{3}} \right )^{\frac{1}{3k - 0.5|V|}} \ge \left (\frac{3b}{2b + a'}\right )^{\frac{\gamma}{5}} 
\end{align*}
The final inequality follows since $k \le |V|$ and $|E| = \nicefrac{3|V|}{2}$. 
Since $a' < b$, this is a positive constant, and we are done.
\end{proof}

\citet{lee:j:apx-hardness} shows \apx-hardness of the MNW problem for general additive valuations using the same min vertex cover problem \citep{chlebik2006minvc}, but their reduced instance is more general and not restricted to three fixed values $a$, $b$, and $c$. 
It is also worth noting that our proof leads to a constant factor lower bound of $1.00013$ which improves on their constant factor lower bound of $1.00008$. 

\begin{restatable}{corr}{corrconstantlowerbound}\label{corr:constant-lower-bound}
Assume agents have $\{0, 1, 3\}$-valuations. It is impossible to approximate MNW by a factor smaller than $1.00013$ unless $\p = \np$.
\end{restatable}
% \begin{proof}
% This proof comes from the observation that the min vertex cover of any $3$-regular graph with $n$ nodes has size upper bounded at $2n/3$. 

% This property follows from the fact that every $3$-regular graph (other than $K_4$) has an independent set of size at least $n/3$ (Brooks' Theorem) \citep{Brooks1941}. The complement of this independent set must be a vertex cover of size at most $2n/3$. 

% We can use this property to upper bound $k$ in the proof of Theorem \ref{thm:mnw-second-apx-hard} with $2n/3$. This gives us the lower bound of 
% \begin{align*}
%     \frac{\NSW^+}{\NSW^{-}} = \left (\left (\frac{3b}{2b + a'}\right )^{\frac{\gamma |E|}{3}} \right )^{\frac{1}{3k - 0.5|V|}} \ge \left (\frac{3b}{2b + a'}\right )^{\frac{\gamma}{3}} 
% \end{align*}

% Finally, applying $a' = 2/3$ (from our choices of $a$, $b$, and $c$) and plugging in $\gamma = \frac{1}{297}$ (from the proof of Lemma \ref{lem:max-3-vc-hardness}) gives us our constant.
% \end{proof}

The vertex cover based proof technique also shows that computing max Nash welfare allocations is \apx-hard even when agents have $\{3, c\}$ valuations where $c > 3$ and is not divisible by $3$; this resolves an open question posed by \citet{akrami2022mnw}.

\begin{restatable}{prop}{propthreeccase}\label{prop:3-c-case}
When agents have $\{3, c\}$-valuations with $c > 3$ and $c$ not divisible by $3$, computing an MNW allocation is \apx-hard.
\end{restatable}

We also show that the techniques used in this section can also be used to show \apx-hardness for computing MEW allocations, when all items have non-negative value. 

\begin{restatable}{theorem}{thmgoodsmewhard}\label{thm:goods-mew-hard}
When agents have $\{a, b, c\}$-valuations with $0 \le a < b < c$, computing an MEW allocation is \apx-hard.
\end{restatable}

Note that MEW allocations can be defined even when agents have negative values for the items; the above proof can be used to show \np-hardness for computing an MEW allocation even when $a$ is negative.

\section{Mixed Manna}\label{sec:chores}\label{sec:mixed-manna}
Next, we consider mixed manna, i.e. the case where agents have $\{a, b, c\}$-valuations with $a < b \le 0 < c$. Note that when two of $a, b$ and $c$ are positive, the problem of computing MEW allocations is NP-hard (Theorem \ref{thm:goods-mew-hard}). For the cases when $a$ and $b$ are negative, we have the following hardness results.

\begin{restatable}{theorem}{thmmixedmannamewhard}\label{thm:mixedmanna-mew-hard}
When agents have $\{a, c\}$-valuations with $a < 0 < c$ and $|a| > |c|$, computing an MEW allocation is \np-hard.  
\end{restatable}

The above result follows from reductions using the decision version of the 2P2N-3SAT problem. We also have the following hardness result from \citet{cousins2023mixedmanna}.

\begin{theorem}[\cite{cousins2023mixedmanna}]\label{thm:cousins2023mixedmanna}
When agents have $\{a, c\}$-valuations with 
\begin{inparaenum}[(i)]
    \item $a < 0 < c$, 
    \item $|a| \ge 3$, and
    \item $|a|$ and $|c|$ are coprime, 
\end{inparaenum}
computing an MEW allocation is \np-hard.  
\end{theorem}

We can use these results to show the following, again using the 2P2N-3SAT problem:
\begin{restatable}{theorem}{thmtwonegative}\label{thm:two-negative}
When agents have $\{a, b, c\}$-valuations with $a < b < 0 < c$, computing an MEW allocation is \np-hard.  
\end{restatable}

Combining \Cref{thm:mixedmanna-mew-hard,thm:cousins2023mixedmanna,thm:two-negative} shows that the only case where one could hope to compute an MEW allocation is when agents have $\{-2, 0, c\}$-valuations. 
We could not show a hardness result for this case and conjecture that it may admit efficient algorithms.
\begin{conjecture}
    There exists an efficient algorithm for computing MEW allocations when agents have $\{-2,0,c\}$-valuations.
\end{conjecture}

For this case of $\{-1, 0, c\}$-valuations (with $c > 1$), an efficient algorithm to compute MEW allocations is known \citep{cousins2023mixedmanna}. The algorithmic results of \citet{cousins2023mixedmanna} apply to a broader class of valuations they define as {\em order-neutral submodular valuations}. This is a more general class than additive valuations but a strict subset of submodular valuations. 
This restriction raises the natural question of whether the restriction from submodular valuations to order-neutral submodular valuations is necessary. 
The answer to this question turns out to be quite surprising. 
However, before we present it, we must first formally define $A$-submodular valuations and order neutrality. 

\begin{definition}[$A$-submodular]
Given a set of integers $A$, A valuation function $v_i$ is $A$-submodular if:
\begin{inparaenum}[(a)]
    \item $v_i(\emptyset) = 0$,
    \item for any $o \in O$ and $S \subseteq O \setminus \{o\}$, $v_i(S \cup \{o\}) - v_i(S) \in A$, and
    \item for any $o \in O$ and $S \subseteq T \subseteq O \setminus \{o\}$, $v_i(S \cup \{o\}) - v_i(S) \ge v_i(T \cup \{o\}) - v_i(T)$.
\end{inparaenum}
In simple words, the valuation function is submodular, and the marginal gains are restricted to values in $A$.
\end{definition}

\begin{definition}[Order Neutrality]
A submodular function $v_i$ is order neutral if for all subsets $S \subseteq O$ and two permutations of the items in $S$, $\pi, \pi':[|S|]\rightarrow S$, the multi-set $\{v_i(\bigcup_{j \in [k]}\pi(j)) - v_i(\bigcup_{j \in [k-1]}\pi(j))\}_{k \in [|S|]}$ is identical to 
the multi-set $\{v_i(\bigcup_{j \in [k]}\pi'(j)) - v_i(\bigcup_{j \in [k-1]}\pi'(j))\}_{k \in [|S|]}$.

In simple words, the order in which the items are added to the set does not affect the marginal gains of the set of items.
\end{definition}

We now present our results. We first show that when agents have $\{-1, 0, 1\}$-submodular valuations, computing an MEW allocation is intractable. 

\begin{restatable}{theorem}{thmsubmodularhard}\label{thm:submodular-hard}
When agents have $\{-1, 0, 1\}$-submodular valuations, computing an MEW allocation is \np-hard.  
\end{restatable}
\begin{proof}
We reduce from the NP-complete restricted exact 3 cover problem~\citep{gon:j:clustering-minimize}.

\prob{Restricted Exact 3 Cover (RX3C)}%
{A finite set of elements $U = \{1, 2, \dots, {3k}\}$, and a collection of 3-element subsets of $U$ (denoted by $\cal F$) such that each element in $U$ appears in exactly $3$ subsets in $\cal F$.}%
{Does there exist a set of triples $\cal F' \subset \cal F$ such that every element in $U$ occurs in exactly one subset in $\cal F'$?}

% The Restricted Exact 3 Cover pris \np-complete even when the graph $G$ is $3$-regular --- that is, each vertex has exactly three edges incident on it. \citet{gon:j:clustering-minimize} shows this result for the closely related problem Exact Cover by 3-Sets where each element appears in exactly three triples and it is easy to see that the analogous result for 3D matching can be shown using this approach (see~\cite{csa:t:popularity-one-sided-matching}).

Note that $|\cal F| = 3k$ since each of the $3k$ elements in $U$ appear in exactly $3$ sets in $\cal F$.
Given an instance of RX3C, we construct a fair allocation instance with $3k$ agents and $9k$ items. 

The $9k$ items are defined as follows: For each element $i \in U$, we have two items $i$ and $i'$ corresponding to the element. We also have $k$ {\em cover} items and $2k$ {\em padding} items. 

The $3k$ agents are defined as follows: For each subset $F = \{i, j, k\}$ in $\cal F$, we have an agent who has the %following
valuation function $v_F$. We describe this valuation function in terms of its marginal gains to make it clear that it is submodular. If the bundle does not contain a cover item, the first item corresponding to the elements $i$, $j$, and $k$ added to the bundle have a marginal value of $0$. The second item adds a marginal value of $-1$. So $v_F(\{i, j, k\}) = 0$ but $v_F(\{i, i', j, k\}) = -1$. This marginal gain of $0$ occurs only if the bundle does not contain a cover item; otherwise the marginal value is $-1$.

The cover items add a marginal value of $1$ when added to an empty bundle. Otherwise they add a marginal value of $0$. All padding items add a marginal value of $1$, irrespective of the bundle they are added to. All other marginal values are $-1$.

If the original RX3C instance admits an exact cover, we can construct an allocation with egalitarian welfare $1$. If the cover consists of the set of triples $\cal F' \subseteq \cal F$, we give all the agents in $\cal F'$ cover items, and all the agents outside $\cal F'$ a padding item along with a copy of each element the subset contains.

If the original RX3C instance does not admit an exact cover, then assume for contradiction that an allocation $X$ achieves an egalitarian welfare of at least $1$. Since there are only $3k$ items that provide a marginal value of $1$, each agent must receive exactly one of these items at the marginal value of $1$ such that no other item in their bundle provides a marginal value of $-1$. 

This implies the set of agents who receive a cover item must not receive any other item. Let this set of agents be $\cal F'$, and since there are $k$ cover items we know $|\cal F'| = k$.
Since $\cal F'$ is not an exact cover, there must be at least one element $i$ present in two subsets in $\cal F'$. This implies at least one of the items corresponding to the element $i$ must be allocated at a marginal value of $-1$. This is a contradiction, giving us our required separation.
\end{proof}

The proof of \Cref{thm:submodular-hard} (or at the very least, this proof technique) does not extend beyond $\{-1, 0, 1\}$-submodular valuations to $\{-1, 0, c\}$-submodular valuations. It turns out, rather surprisingly, that MEW allocations can be computed efficiently when agents have $\{-1, 0, c\}$-submodular valuations with $c \ge 2$. We show this by proving that all $\{-1, 0, c\}$-submodular valuations are order neutral, thereby showing they fall under the class of valuations for which \citet{cousins2023mixedmanna} present an efficient algorithm.

\begin{restatable}{prop}{propsubmodularorderneutral}\label{prop:submodular-orderneutral}
When $c \ge 2$, all $\{-1, 0, c\}$-submodular valuations are order neutral.   
\end{restatable}
\begin{proof}
Let $v_i$ be a $\{-1, 0, c\}$ submodular valuation for $c \ge 2$.
Consider some bundle, some order $\pi$ over the items in the bundle, and some items $o, o'$ which appear consecutively in the order $\pi$. That is, $\pi$ consists of the set of items $S$ (in some order) followed by the items $o$ and $o'$ followed by another set of items $S'$ (in some order).

If we swap $o$ and $o'$, exactly two marginal values change. More specifically, $v_i(S + o) - v_i(S)$ and $v_i(S + o + o') - v_i(S + o)$ become $v_i(S + o') - v_i(S)$ and $v_i(S + o + o') - v_i(S + o')$. The value of the bundle remains the same no matter which order we use, so we must have
\begin{align*}
	&\left(v_i(S + o') - v_i(S)\right) + \left(v_i(S + o + o') - v_i(S + o')\right) \\&\quad = \left(v_i(S + o) - v_i(S)\right) + \left(v_i(S + o + o') - v_i(S + o)\right)
\end{align*}

The statement follows from noting that when $c \ge 2$, every value of $v_i(S + o + o') - v_i(S)$ has a unique decomposition into two values. 
More specifically, the value of $v_i(S + o + o') - v_i(S)$ can only be $-2,-1,0,c-1,c$ or $2c$. In each case the marginal gains of the two items are encoded by exactly the same two values. 
For example, the only possible way $v_i(S + o + o') - v_i(S) = c-1$ is if one of the items provides a value of $c$ and the other provides a value of $-1$; if it is $-2$ then both items offer a marginal gain of $-1$. 
Therefore, the set $\{v_i(S + o) - v_i(S), v_i(S + o + o') - v_i(S + o)\}$ must be exactly equivalent to the set $\{v_i(S + o') - v_i(S), v_i(S + o + o') - v_i(S + o')\}$ if they sum up to the same value. 
This implies that swapping two consecutive elements in an order retains order neutrality. 

Since we can move from any order $\pi$ to any order $\pi'$ using consecutive element swaps (as is done in bubble sort), $v_i$ must be order neutral.

This proves the statement. We note interestingly that this argument does not hold for $\{-1, 0, 1\}$ submodular valuations, since $\{-1, 1\}$ and the set $\{0, 0\}$ have the same sum. So if $\{v_i(S + o) - v_i(S), v_i(S + o + o') - v_i(S + o)\} = \{-1, 1\}$, swapping the two items could lead to the set of marginal gains $\{0, 0\}$. We use this specific property to show \np-hardness for $\{-1, 0, 1\}$ valuations in Theorem~\ref{thm:submodular-hard}.
\end{proof}

\section{Conclusions and Future Work}
In this work, we almost completely characterize the complexity of computing max Nash welfare and max egalitarian welfare allocations under ternary valuations. 
Rather unfortunately, we show that existing algorithms that work under binary and bivalued valuations cannot be generalized beyond bivalued valuations. 
Specifically, our results highlight a fundamental limitation of the path augmentation technique used heavily to design algorithms for binary and bivalued valuations \citep{barman2018pathtransfers, akrami2022mnw, Barman2021MRFMaxmin, barman2022groupstrategyproof,viswanathan2022yankee,viswanathan2022generalyankee}.

There are two natural questions left for future work. The first is the complexity of computing MEW allocations under $\{-2, 0, c\}$ valuations. Resolving this question would complete our characterization. The second question is to gain a further understanding of the increase in hardness as we generalize beyond additive and into submodular valuations. We know from Theorem \ref{thm:submodular-hard} that some problems become significantly harder as we move from additive to submodular valuations but the results of \citet{Babaioff2021Dichotomous} suggests that some problems still remain easy. There are two specific cases whose complexity still remain open questions. The first is computing MEW allocations under $\{-2, 0, c\}$ submodular valuations and the second is computing MNW allocations under $\{2, c\}$ submodular valuations. Resolving these two cases would result in a complete characterization of max Nash welfare and max egalitarian welfare allocations under submodular valuations as well.

\bigskip

\noindent
{\bf Acknowledgments:} Research done while Fitzsimmons was on sabbatical visiting the University of Massachusetts, Amherst. Viswanathan and Zick are supported by an NSF grant CISE:IIS:RI:\#2327057. Fitzsimmons is supported in part by NSF grant CCF-2421978.

\newpage

\bibliography{abb,literature}

\newpage

\appendix

\section{Other Proofs from Section \ref{sec:goods}}\label{apdx:goods}

\lemmaxthreevchardness*
\begin{proof}
We prove this using our problem's relation to the minimum vertex cover problem. Given a graph $G$, the minimum vertex cover problem asks for a minimum sized set of vertices such that each edge is incident to a vertex in the set.
%the minimum vertex cover of the graph $G$ is. 
\cite{chlebik2006minvc} show that the problem of computing a minimum vertex cover cannot be approximated by a factor of $\frac{100}{99}$ even when the graph $G$ is $3$-regular. 

In their reduction, they construct a $3$-regular graph $G = (V, E)$ and show that deciding whether the graph has a vertex cover of size $k^*$ (for some $k^*$) or whether all vertex covers of the graph have size at least $\frac{100}{99}k^*$ is \np-hard. The exact value of $k^*$ is unimportant; one only needs to note that $k^* \ge \frac{|V|}{2}$ since we need at least $\frac{|V|}{2}$ nodes in a $3$-regular graph to cover $|E| = \frac{3|V|}{2}$ edges. 

This can be easily rephrased as our problem. We first must note that if all vertex covers have size at least $\frac{100}{99}k^*$, then any subset of nodes of size $k^*$ must not cover at least $\frac{1}{99}k^*$ edges; if there is a subset of nodes of size $k^*$ that covers strictly more than $|E| - \frac{1}{99}k^*$ edges, we can use it to trivially construct a vertex cover of size strictly less than $\frac{100}{99}k^*$. 

Therefore, given a $3$-regular graph $G = (V, E)$ and $k^*$, it is \np-hard to decide whether the graph $G$ has a vertex cover of size $k^*$ or whether all $k^*$ sized subsets of $|V|$ must not cover at least $\frac{1}{99}k^*$ edges. The exact phrasing of the lemma comes from lower bound $\frac{1}{99}k^*$ as 
\begin{align*}
    \frac{1}{99}k^* \ge \frac{1}{198}|V| = \frac{1}{297}|E|
\end{align*}
The final equality comes from the fact that the graph is $3$-regular and therefore $|E| = \frac{3|V|}{2}$.
\end{proof}

\corrconstantlowerbound*
\begin{proof}
This proof comes from the observation that the min vertex cover of any $3$-regular graph with $n$ nodes has size upper bounded at $2n/3$. 

This property follows from the fact that every $3$-regular graph (other than $K_4$) has an independent set of size at least $n/3$ (Brooks' Theorem \citep{Brooks1941}). The complement of this independent set must be a vertex cover of size at most $2n/3$. 

We can use this property to upper bound $k$ in the proof of Theorem \ref{thm:mnw-second-apx-hard} with $2n/3$. This gives us the lower bound of 
\begin{align*}
    \frac{\NSW^+}{\NSW^{-}} = \left (\left (\frac{3b}{2b + a'}\right )^{\frac{\gamma |E|}{3}} \right )^{\frac{1}{3k - 0.5|V|}} \ge \left (\frac{3b}{2b + a'}\right )^{\frac{\gamma}{3}} 
\end{align*}

Finally, applying $a' = 2/3$ (from our choices of $a$, $b$, and $c$) and plugging in $\gamma = \frac{1}{297}$ (from the proof of Lemma \ref{lem:max-3-vc-hardness}) gives us our constant.
\end{proof}
\propthreeccase*
\begin{proof}
This proof follows from a reduction from Lemma \ref{lem:max-3-vc-hardness} and is very similar to Theorem \ref{thm:mnw-second-apx-hard}.

Given a $3$-regular graph $G = (V, E)$ and an integer $k$, we construct a fair allocation instance with $6k + ck - 1.5|V|$ items and $3k - 0.5|V|$ agents.

The $6k + ck - 1.5|V|$ items are defined as follows:
\begin{enumerate}[(a)]
    \item For each each edge $e_{ij} \in E$ we have an item $e_{ij}$,
    \item We have $ck$ {\em vertex cover} items $c_1, \dots, c_k$, and 
    \item We have $6k - 3|V|$ special items.
\end{enumerate}

The valuation function of the $3k - 0.5|V|$ agents are defined as follows:
\begin{enumerate}[(a)]
    \item For each node $i \in V$, we have an agent who values the edges incident on it at $c$, and
    \item We have $3k - 1.5|V|$ dummy agents, who value exactly two special items each at $c$ and value all the edge items at $c$. 
\end{enumerate}

All unmentioned values are at $3$. We also ensure that no two dummy agents have any overlap in the special items they value at $c$.

We now prove the correctness of our reduction.

($\Longrightarrow$) Assume the graph admits a vertex cover (say $S$) of size $k$. Then, we allocate the $c$ vertex cover items to each of the agents in $S$. For all other nodes in the graph, we allocate the items corresponding to the three edges they are incident on. This is possible since at least one endpoint of each edge belongs to the vertex cover $S$.

Note that we have only allocated $ck + 3|V| - 3k$ items. We allocate the remaining items to the dummy agents. For each dummy agent, we allocate the two special items they value at $c$ and any one unallocated edge item. This way, we allocate all $6k + ck - 1.5|V|$ items.

All agents receive a utility of $3c$. The Nash welfare of this allocation is 
\begin{align*}
    \NSW(v, X) = 3c && (\NSW^+)
\end{align*}
This allocation also maximizes utilitarian social welfare. Additionally, since it maximizes utilitarian social welfare and gives all agents the exact same utility, this is the highest possible Nash welfare achievable on this instance.

($\Longleftarrow$) Assume that no subset of nodes of size $k$ covers more than $(1-\gamma)|E|$ edges. Consider the subset of nodes (say $S$) who receive $c$ vertex cover items that they value at $3$ each. This subset of nodes must not cover at least $\gamma|E|$ edges.
For each of these edges, one of its endpoints must not receive a utility of exactly $3c$. This is because both endpoints only value three items at $c$ and at most one of these agents can receive the edge between them which is not in the vertex cover. Note that is argument uses the fact that $c$ is not divisible by $3$, so if an agent receives a utility of exactly $3c$, they must either receive $3$ items valued at $c$ or $c$ items valued at $3$.

If we do this for all edges, we get that there at least $\frac{\gamma|E|}{3}$ agents who do not receive a utility of exactly $3c$. We divide by three since $G$ is $3$-regular, so each node is counted at most thrice.

Assume some set of agents $N'$ ($|N'| = m'$) do not receive a utility of exactly $3c$. Since utilities are integral, these agents must receive a utility of at most $3c-1$ or at least $3c+1$. To maximize Nash welfare these agents must have utilities as close to $3c$ as possible and must have a utilitarian social welfare upper bounded by $3c(3k - 0.5|V|)$.
This gives us the max Nash welfare upper bound of
\begin{align*}
    &\left (3c^{3k - 0.5|V| - m'} \left( (3c - 1)(3c + 1)\right)^{m'/2}\right )^\frac{1}{3k - 0.5|V|} \\ &\quad \le \left (3c^{3k - 0.5|V|} \left( \frac{(3c - 1)(3c + 1)}{9c^2}\right)^{\frac{\gamma|E|}{6}}\right )^\frac{1}{3k - 0.5|V|} && (\NSW^{-})
\end{align*}
The last inequality follows from using $m' \ge \frac{\gamma |E|}{3}$.
We have shown that it is \np-hard to decide whether an allocation has Nash welfare at least $\NSW^+$ or whether all allocations have Nash welfare at most $\NSW^{-}$. Taking the ratio of the two Nash welfares gives us the following approximation lower bound
\begin{align*}
    \frac{\NSW^+}{\NSW^{-}} = \left (\left (\frac{9c^2}{9c^2 - 1}\right )^{\frac{\gamma |E|}{6}} \right )^{\frac{1}{3k - 0.5|V|}} \ge \left (\frac{9c^2}{9c^2 - 1}\right )^{\frac{\gamma}{10}} 
\end{align*}
The final inequality follows from upper bounding $k$ at $|V|$ and substituting $|E| = 3|V|/2$. Since this is a positive constant, we are done.
\end{proof}

\thmgoodsmewhard*
\begin{proof}
We reduce from the decision version of the 2P2N-3SAT problem and use a construction very similar to Theorem \ref{thm:mnw-first-apx-hard}. Our proof is divided into three cases to account for the different possible values of $a$, $b$, and $c$. All three cases use a similar construction with the second and third cases following from minor modifications to the first case.

\noindent\textbf{Case 1: $\bm{c \ge 2b}$.} 

% \prob{2P2N-3SAT}%
% {A boolean formula $\phi(x_1, \dots, x_n)$ where $\phi$ is
% in 3CNF and each variable $x_i$ occurs in $\phi$ exactly twice as positive literal and twice as a negated literal.}%
% {Does there exist an assignment to the variables $x_1, \dots, x_n$ such that $\phi$ is true?}

Let $\phi(x_1, \dots, x_n)$ be an instance of 2P2N-3SAT with $m$ clauses $(C_1, \dots, C_m)$. We construct an
instance of our allocation problem with $4n$ agents and $7n$ items as follows.

% Items.
% We first describe the items in our instance.
For each variable $x_i$ we have five items: $x_i$ and $x'_i$ corresponding to the positive literals, $\overline x_i$ and  $\overline x'_i$ corresponding to the negative literals, and another item $clog_i$ which we will call the {\em sink clogger} item. In addition to these variable items, we have 
$2n$ {\em special} items.

% Agents
We now describe the agents in our instance.
\begin{itemize}
\item For each variable $x_i$ we have two agents $pos_i$ and $neg_i$.
The $pos_i$ agent values $x_i$ and $x'_i$ (the positive literals) at $b$, and values $clog_i$ at $c$. 
Similarly, the $neg_i$ agent values $\overline x_i$ and $\overline x'_i$ (the negative literals) at $b$, and values  $clog_i$ at $c$. We refer to the agents $pos_i$ and $neg_i$ as {\em sink} agents and, as mentioned before, the item $clog_i$ as a sink clogger item.

\item For each clause $C_i$, we have a clause agent who values all copies of the items corresponding to the three literals in the clause $C_i$ at $b$, and exactly one special item at $b$. 

\item We have $2n - m$ dummy agents, who value all copies of all the literals at $b$ and exactly one special item at $b$. 

All unmentioned values are $a$. We set up these values in such a way that each special item is valued by exactly one agent (either a dummy agent or a clause agent) at $b$. 
\end{itemize}

% The main idea of this proof is that, in any MNW allocation, one of the sink agents must be clogged by its corresponding sink clogger. This clogging can be intuitively understood as choosing that literal form in the 3SAT instance. For example, if $pos_i$ is allocated the sink clogger $clog_i$, then the items $x_i$ and $x'_i$ can be allocated to the clause agents. Perhaps not obviously so, if $pos_i$ is clogged by the sink clogger, then $neg_i$ must be allocated the items $\overline{x_i}$ and $\overline{x'_i}$; this ensures consistency in the allocation to the clause agents.

Assume that there exists a satisfying assignment $\sigma$ to the 3SAT instance. Then, we argue that our constructed instance will have a max egalitarian welfare of at least $2b$. There is a straightforward assignment that achieves this. Pick any satisfying assignment to the 3SAT instance. If the assignment assigns $\sigma(x_i) = 1$, then we allocate $clog_i$ to $pos_i$, and $\overline x_i$ and $\overline x'_i$ to $neg_i$. Otherwise if $\sigma(x_i) = 0$, then we allocate $clog_i$ to $neg_i$, and $x_i$ and $x'_i$ to $pos_i$. We do this for each $x_i$ to decide the allocation to each $pos_i$ and $neg_i$.

The clause agents receive their corresponding special item that they value at $b$ and exactly one item corresponding to a copy of a literal that satisfies the clause in the assignment $\sigma$ (the choice can be made arbitrarily). 

Finally, for each of the dummy agents, we allocate their corresponding special item and exactly one item corresponding to a literal that has not yet been allocated; again, this choice can be made arbitrarily.

% Such an allocation is possible since after the allocation to the sink agents, there are exactly $2n$ literal items left and there are exactly $2n$ clause and dummy agents. 
It is easy to see that this allocation achieves an egalitarian welfare of $2b$ --- each sink agent clogged by the sink clogger receives a utility of $c$ and all other agents receive a utility of $2b$.

Assume the original 3SAT instance does not admit a satisfying assignment. Consider the max egalitarian allocation $X$ in our constructed instance. Our goal is to show that the egalitarian welfare of $X$ is strictly less than $2b$.

In this allocation $X$, we can assume without loss of generality that the sink clogger items are allocated to the sink agents. If this is not the case, we can swap a sink clogger item $clog_i$ with one of the items allocated to the sink agents $pos_i$ and $neg_i$ to Pareto improve the allocation. Note that $pos_i$ and $neg_i$ must receive at least one item in $X$ since otherwise the egalitarian welfare will be $0$.

Assume $X$ has an egalitarian welfare of at least $2b$. 
These $n$ agents who receive the sink clogger item receive a utility of $c \ge 2b$. There are at most $6n$ items left to be allocated among the remaining $3n$ agents. Since the sink clogger items are the only ones to be valued at $c$, for the remaining $3n$ agents to receive a utility of at least $2b$, all of them must receive exactly $2$ items each that they value at $b$.

This means, for each variable $x_i$, the sink agent that does not receive a sink clogger item must receive both items that they value at $b$. That is, if $clog_i$ is allocated to $pos_i$, then $neg_i$ must receive $\overline x_i$ and $\overline x'_i$. If this does not happen, then $neg_i$ receives two other items that give it a utility of at most $b+a$ violating our initial assumption.

% This allocation to the sink agents enforces consistency in the allocation to
If the allocation satisfies the above property, then the allocation to the clause agents must be consistent --- that is, the literal items allocated to the clauses must either be positive or negative for a particular variable $x_i$ but cannot be both. Since the input 3SAT is unsatisfiable, at least one clause cannot receive a second item with utility $b$; therefore if they receive two items, they receive a utility of at most $b + a$. This contradicts our initial assumption that the allocation $X$ has an egalitarian welfare of at least $2b$.

To conclude, when the input 3SAT instance is satisfiable, the max egalitarian welfare is at least $2b$; otherwise, it is at most $b + a$. The ratio $\frac{2b}{b+ a}$ is a constant, so we have a constant factor approximation lower bound. This completes the proof for this particular case.

\noindent \textbf{Case 2: $\bm{2b > c}$ and $\bm{a = 0}$.} This case follows from exact same reduction and proof as the previous case. However, in this case, when there is a solution to the 3SAT instance, the egalitarian welfare is at least $c$ since $2b > c$. When there is no solution to the 3SAT instance, the egalitarian welfare is at most $b + a = b$. This gives us our constant factor lower bound.

\noindent \textbf{Case 3: $\bm{2b > c}$ and $\bm{a > 0}$.} For our final case, we modify the instance slightly and add $n$ padding items that all agents value at $a$. 

For this new instance, when the original $3$SAT instance is satisfiable, we can find an allocation that achieves an egalitarian welfare of $\min\{c+a, 2b\}$. This can be constructed easily using the allocation from Case 1 and giving the padding items to the $n$ agents who receive a sink clogger item. 

Note that to achieve an egalitarian welfare of $\min\{c+a, 2b\}$, every agent must receive at least two items. Since there are exactly double the number of items as there are agents, every agent must receive exactly two items to achieve an egalitarian welfare of $\min\{c+a, 2b\}$. 

The rest of the proof flows very similarly to Case 1.
Assume the original $3$SAT instance does not admit a satisfying assignment, but there is an allocation $X$ that achieves an egalitarian welfare of at least $\min\{c+a, 2b\}$. 
We can show using an argument similar to Case 1 that the allocation $X$ must allocate all sink clogger items to sink agents. 

Then we can similarly show that if $pos_i$ receives the sink clogger item, then $neg_i$ must receive the items $\overline x_i$ and $\overline x'_i$, which enforces consistency in the allocation to the clauses. 

Finally, we can use the fact that there is no satisfying assignment in the original 3SAT instance to show that at least one of the clause agents cannot receive two items that they value at $b$. Therefore $X$ cannot have an egalitarian welfare of $\min\{c+a, 2b\}$. This gives us a constant factor lower bound of $\frac{\min\{c+a, 2b\}}{\min\{c+a, 2b\}-1}$.
\end{proof}

% \propgoodsapx*
% \begin{proof}
% If $a = 0$, we transform this instance into an instance with binary valuations where all items valued at $b$ or $c$ are valued instead at $1$. Then we compute a leximin allocation \citep{halpern2020binaryadditive,barman2018pathtransfers} with these binary valuations. This allocation is trivially a $c/b$ approximation to the max egalitarian welfare.

% If $a > 0$, the algorithm is even simpler. We divide the set of items into $n$ parts such that two parts have a size difference of at most $1$. We then give each agent an arbitrary part. This allocation is trivially a $c/a$ approximation to the max egalitarian welfare.
% \end{proof}

\section{Missing Proofs from Section \ref{sec:chores}}
% \propchoresapx*
% \begin{proof}
% We use the same algorithm from the previous proof. We divide the set of items into $n$ parts such that two parts have a size difference of at most $1$. We then give each agent an arbitrary part. This allocation is trivially a $\frac ba$ approximation to the max egalitarian welfare.
% \end{proof}

% \section{Missing Proofs from Section \ref{sec:mixed-manna}}

\thmmixedmannamewhard*
\begin{proof}
This proof, like most proofs before this, reduces from the decision version of the 2P2N-3SAT problem. We assume without loss of generality that $|a|$ and $|c|$ are co-prime; otherwise we can scale agent valuations by the greatest common divisor of $|a|$ and $|c|$. 

Let $\phi(x_1, \dots, x_n)$ be an instance of 2P2N-3SAT with $m$ clauses $(C_1, \dots, C_m)$. We construct an instance of our allocation problem with $4n$ agents and $3cn + 3|a|n$ items as follows.

% Items.
% We first describe the items in our instance.
For each variable $x_i$ we have four items: $x_i$ and $x'_i$ corresponding to the positive literals, and $\overline x_i$ and  $\overline x'_i$ corresponding to the negative literals. In addition to these variable items, we have 
$3|a|n - 4n$ {\em special} items and $3cn$ {\em padding} items.

% Agents
We now describe the agents in our instance.
\begin{itemize}
\item For each variable $x_i$ we have two agents $pos_i$ and $neg_i$.
The $pos_i$ agent values $x_i$ and $x'_i$ (the positive literals) at $c$. 
% \zfnote{Need to number these dummies or mention above that each agent has a unique one.}
Similarly, the $neg_i$ agent values $\overline x_i$ and $\overline x'_i$ (the negative literals) at $c$. Further, for each $i$, there are $|a| - 2$ special items that both $pos_i$ and $neg_i$ value at $c$. No other agents value these $|a| - 2$ special items at $c$. Similar to the previous proofs, we call the agents $pos_i$ and $neg_i$ as {\em sink} agents.

\item For each clause $C_i$, we have a clause agent who values all copies of the items corresponding to the three literals in the clause $C_i$ at $c$, and exactly $|a| - 1$ special items at $c$. 

\item We have $2n - m$ dummy agents, who value all copies of all the literals at $c$ and exactly $|a| - 1$ special items at $c$. 

All unmentioned values are $a$. We set up these values in such a way that clause and dummy agents share no overlap in the special items they value at $c$. Note that the $3cn$ padding items are valued by all agents at $a$.
\end{itemize}

Assume that there exists a satisfying assignment $\sigma$ to the $3$-SAT instance. Then, we argue that our constructed instance will have a max egalitarian welfare of at least $0$. There is a straightforward assignment that achieves this. Pick any satisfying assignment to the 3-SAT instance. If the assignment assigns $\sigma(x_i) = 1$, then we allocate $c$ padding items, $\overline x_i$ and $\overline x'_i$ to $neg_i$. We also allocate the $|a| - 2$ special items that $neg_i$ values at $c$ to $neg_i$.

Otherwise if $\sigma(x_i) = 0$, then we allocate $c$ padding items, $x_i$ and $x'_i$ to $pos_i$. We also allocate the $|a| - 2$ special items that $pos_i$ values at $c$ to $pos_i$.
We do this for each $x_i$ to decide the allocation to each $pos_i$ and $neg_i$. Note that for each $i$, one of $pos_i$ or $neg_i$ receive an empty bundle.

The clause agents receive their corresponding special items that they value at $c$, $c$ padding items that they value at $a$ and exactly one item corresponding to a copy of a literal that satisfies the clause (the choice can be made arbitrarily).

Finally, for each of the dummy agents, we allocate them their corresponding special items they value at $c$, $c$ padding items they value at $a$ and exactly one item corresponding to a literal that has not yet been allocated; again, this choice can be made arbitrarily. 

In this allocation, agents either receive a utility of $|a|c + ac = 0$ or $0$ (from an empty bundle). This allocation therefore has an egalitarian welfare of $0$.

Assume the original $3$-SAT instance does not admit a satisfying assignment. Consider the max egalitarian allocation $X$ in our constructed instance. Our goal is to show that the egalitarian welfare of $X$ is strictly less than $0$.

Assume for contradiction that $X$ has an egalitarian welfare of at least $0$. Note that the highest possible utilitarian welfare achievable is exactly $0$, so for an allocation to have egalitarian welfare $0$, all agents must receive the utility $0$. This also implies that to achieve an egalitarian welfare of $0$, an allocation must maximize utilitarian welfare and all items which are valued at $c$ by some agent must provide value $c$ to the agent it is allocated to; this comprises of all the special items and the literal items. 

Therefore, all the special items that clause agents uniquely value at $c$ must be allocated to them. Additionally, at least one of $pos_i$ and $neg_i$ for each $i$ must receive some special items that they value at $c$.

Furthermore, for agents to receive a utility of exactly $0$, they must either receive an empty bundle, or must receive $\alpha c$ items with value $a$ and $\alpha |a|$ items with value $c$ (for some positive integer $\alpha$). This follows from the fact that $a$ and $c$ are assumed to be coprime. 

Since each clause and dummy agent receives some items that give them positive utility, they must receive at least $c$ padding items. Using the same argument, at least one of $pos_i$ and $neg_i$ for each $i$ must receive $c$ padding items. 

Note that there are $3cn$ padding items, so exactly one of $pos_i$ and $neg_i$ for each $i$ must receive $c$ padding items.

% The remaining padding items behave like sink clogger items (from the previous proofs), and must be allocated to the sink agents. Moreover, exactly one agent out of $pos_i$ and $neg_i$ (for some $i$) must receive a padding item along with all the special items they value at $c$. Otherwise, some agent will have a negative utility. 

The rest of the proof follows similarly to all the 2P2N-3SAT proofs in this paper. If $c$ padding items is allocated to $pos_i$, then $pos_i$ must receive $x_i$, $x'_i$ and $|a| - 2$ special items to ensure their utility is $0$. This enforces consistency in the allocation to the clauses.

Since the input $3$-SAT is unsatisfiable, at least one clause cannot receive a literal item they value at $c$. This means that clause agent must have negative utility, contradicting our initial assumption about $X$.

To conclude, when the input 3-SAT instance is satisfiable, the max egalitarian welfare is at least $0$; otherwise, it is negative. This gives us the required separation.
\end{proof}

\thmtwonegative*
\begin{proof}
This proof, again, reduces from the decision version of the 2P2N-3SAT problem similar to Theorems \ref{thm:mnw-first-apx-hard} and \ref{thm:goods-mew-hard}.

If $a \ne 2b$, the problem is NP-hard \citep{lenstra1990chores}. Therefore we can assume $a = 2b$. If $c < -a$, then computing MEW allocations with $\{a, c\}$ valuations is NP-hard (Theorem \ref{thm:mixedmanna-mew-hard}). 
We assume that $c \ge -a$ and let $d$ be the largest common factor of $-a$ and $c$; in particular $\frac ad$ and $\frac cd$ are co-prime by definition. 
If $\frac ad$ is either $1$ or $2$ (i.e. $a = d$ or $a = 2d$) then $c$ is a multiple of $b = \frac a2$, since then $b = \frac d2$ or $b = d$. 
If $\frac ad \ge 3$, we apply Theorem \ref{thm:cousins2023mixedmanna} which shows that the problem of computing MEW allocations is NP-hard when agents have $\left\{\frac ad, \frac cd\right\}$ valuations. 
By simply scaling the valuations, this result also implies that computing MEW allocations is NP-hard when agents have $\{a, c\}$ valuations.

Therefore, the only case we need to examine is when $a = 2b$ and $c = -k^* b$ for some $k^* \ge 2$. 
To show that this case is hard, we reduce from the 2P2N-3SAT problem.

Let $\phi(x_1, \dots, x_n)$ be an instance of 2P2N-3SAT with $m$ clauses $(C_1, \dots, C_m)$.
We construct an instance of our allocation problem with $4n$ agents and $3nk^*+3n$ items as follows.

% Items.
% We first describe the items in our instance.
For each variable $x_i$ we have four items: $x_i$ and $x'_i$ corresponding to the positive literals, and $\overline x_i$ and  $\overline x'_i$ corresponding to the negative literals. In addition to these variable items, we have $3nk^* - 4n$ {\em special} items and $3n$ {\em padding} items.

% For each pair of sink agents $\Pos_i,\Neg_i$, we have a special bundle $S_i$ that consists of $k^*-2$ items. 
% In addition, for each clause agent and dummy agent, there is a special bundle of $k^*-1$ items. 
% Thus, there is a total of $3nk^*-4n$ special items.
% There is no overlap in the special bundles. 

% Agents
We now describe the agents in our instance.
\begin{itemize}
% \item Each agent values their set of special items at $b$. Thus, $\Pos_i$ and $\Neg_i$ value every item in $S_i$ at $b$, and each clause/dummy agent values every item in their (unique) bundle of $k^*-1$ items at $b$. No other agents value items in the special bundles at $b$ each
\item For each variable $x_i$ we have two agents $\Pos_i$ and $\Neg_i$.
The $\Pos_i$ agent values $x_i$ and $x'_i$ (the positive literals) at $b$.
Similarly, the $\Neg_i$ agent values $\overline x_i$ and $\overline x'_i$ (the negative literals) at $b$. 
Furthermore, for each $i$, there is a bundle of $k^* - 2$ special items that both $\Pos_i$ and $\Neg_i$ value at $b$ each. 
No other agents value these $k^* - 2$ special items at $b$. 
Finally, there is exactly one padding item that both $\Pos_i$ and $\Neg_i$ value at $c$, and no other agent values this padding item at $c$. 
Similar to the previous proofs, we refer to the agents $\Pos_i$ and $\Neg_i$ as {\em sink} agents.
\item For each clause $C_i$, we have a clause agent who values all copies of the items corresponding to the three literals in the clause $C_i$ at $b$. 
The agent $C_i$ also values a bundle of $k^* - 1$ special items at $b$ each and exactly one padding item at $c$.  
\item We have $2n - m$ dummy agents, who value all copies of all the literals at $b$. 
The dummy agent also values a bundle of $k^* - 1$ special items at $b$ each and a unique padding item at $c$. 

All unspecified values are $a$. 
We set up these values in such a way that clause and dummy agents share no overlap in the bundle of special items they value at $b$ or the padding items they value at $c$. 
\end{itemize}
We argue that if there exists a satisfying assignment to the 3SAT instance, then our constructed instance will have a max egalitarian welfare of $0$. 
There is a straightforward assignment that achieves this. 

Pick any satisfying assignment $\sigma$ to the 3SAT instance. If the assignment assigns $\sigma(x_i) = 1$, then we allocate $\overline x_i$ and $\overline x'_i$ to $\Neg_i$. 
We also allocate to $\Neg_i$ the one padding item they value at $c$ and the $k^* - 2$ special items they value at $b$.
The total utility of $\Neg_i$ is thus $2b + c + (k^* - 2)b = k^*b + c = 0$
Similarly if $\sigma(x_i) = 0$, then we allocate $x_i$ and $x'_i$ to $\Pos_i$ along with the one padding item $\Pos_i$ values at $c$ and the $k^*-2$ special items $\Pos_i$ values at $b$. 
Similarly, in that case the total utility of $\Pos_i$ is $0$.
We do this for each $x_i$ to decide the allocation to each $\Pos_i$ and $\Neg_i$. 

For the clause and dummy agents, we first allocate all the special items to the agents who {\em uniquely} value them at $b$ and all the padding items to the agents who {\em uniquely} value them at $c$. At this stage, the clause and dummy agents have a utility of $(k^* - 1)b + c = -b$. 

The clause agents receive exactly one item corresponding to a copy of a literal that satisfies the clause in the assignment $\sigma$ (the choice can be made arbitrarily). Since $x$ is a satisfying assignment, such a copy is guaranteed to exist. 
Since the clause agents value this literal item at $b$, this brings their utility from $-b$ to $0$

Finally, for each of the dummy agents, we allocate exactly one item corresponding to a literal that has not yet been allocated; again, this choice can be made arbitrarily, and adds an additional $b$ to their utility, setting it to $0$. 

In this allocation, all agents receive either an empty bundle (if they are the unassigned sink agent), or a padding item and $k^*$ items they value at $b$ for a total utility of $0$. Therefore, this allocation has an egalitarian welfare of $0$. 

Assume the original $3$-SAT instance does not admit a satisfying assignment. Consider the max egalitarian allocation $X$ in our constructed instance. 
Our goal is to show that the egalitarian welfare of $X$ is strictly less than $0$.

Assume for contradiction that $X$ has an egalitarian welfare of at least $0$. 
Since the instance is constructed in a way such that the maximum utilitarian social welfare possible is $0$, 
the allocation $X$ must also maximize utilitarian social welfare. 
This implies that no agent in $X$ receives an item they value at $a$. 

The rest of the proof follows similarly to all the 2P2N-3SAT proofs in this paper. If a padding item is allocated to $\Pos_i$ in $X$, then $\Pos_i$ must receive $x_i$, $x'_i$ along with the $k^* -2$ items that $\Pos_i$ values at $b$; these are the only items that $\Pos_i$ values at $b$. 
Otherwise, $\Pos_i$ will have positive utility which means another agent must have negative utility since the max utilitarian social welfare possible is $0$. We can make a similar argument with $\Neg_i$. 
This enforces consistency in the allocation to the clauses.
That is, given a variable $x_i$, it cannot be the case that some clause agent receives one positive literal (e.g. one of $x_i$ and $x_i'$) and another clause agent receives a negative literal (one of $\overline x_i$ and $\overline x_i'$). 

Since the input 3SAT is unsatisfiable, at least one clause agent cannot receive a literal item they value at $b$ because each of these items have been allocated to the corresponding sink agents. This clause agent must receive the special items and padding items they uniquely value at $c$ and $b$ respectively to ensure this allocation is utilitarian welfare maximizing. If they do not receive any other items, their utility must be positive resulting in the negative egalitarian welfare of the allocation $X$ (because the max social welfare possible is $0$ which immediately implies that \emph{some} agent has a negative utility), or they must receive an item they value at $a$, again resulting in a negative egalitarian welfare. 

To conclude, when the input 3-SAT instance is satisfiable, the max egalitarian welfare is $0$; otherwise, it is at most $-1$. This gives us the required separation.
\end{proof}

\end{document}